\documentclass[letterpaper, 10 pt, conference]{ieeeconf}  % Comment this line out if you need a4paper
\usepackage[left=0.75in, right=0.75in, top=0.75in, bottom=0.8in]{geometry}

\IEEEoverridecommandlockouts                              % This command is only needed if 
\usepackage{acronym}
\usepackage{algorithm}
\usepackage{algpseudocode}
\usepackage{amsmath} % assumes amsmath package installed
\usepackage{amssymb}  % assumes amsmath package installed
\usepackage{anyfontsize}
\usepackage{arydshln}

\usepackage{bm}

\usepackage{cases}
\usepackage{cite}
\usepackage{color}
\usepackage{csquotes}
\usepackage{caption}

\usepackage{epsfig} % for postscript graphics files
\usepackage{etoolbox}
\usepackage{eucal}

\usepackage{graphics} % for pdf, bitmapped graphics files
\usepackage{graphicx}
\usepackage{grffile}

\usepackage{hyperref}

\usepackage{indentfirst}

\usepackage{lipsum}
\usepackage{listings}

\usepackage{mathptmx} % assumes new font selection scheme installed
\usepackage{mathrsfs}
\usepackage{mathtools}

\usepackage{pgfplots}

\usepackage{relsize}

\usepackage{soul}
\usepackage{subcaption}

\usepackage{times} % assumes new font selection scheme installed
\usepackage{tikz}

\usepackage[nameinlink]{cleveref} 

\pgfplotsset{compat=1.16}

\creflabelformat{equation}{#2#1#3}
\crefname{equation}{Equation}{Equations}
\crefname{theorem}{Theorem}{Theorems}
\crefname{definition}{Definition}{Definitions}
\crefname{corollary}{Corollary}{Corollaries}
\crefname{lemma}{Lemma}{Lemmas}

\crefname{section}{Section}{Sections}

\crefname{table}{Table}{Tables}
\crefname{figure}{Figure}{Figures}
\crefname{algorithm}{Algorithm}{Algorithms}

\Crefname{table}{Table}{Tables}
\Crefname{figure}{Figure}{Figures}
\Crefname{algorithm}{Algorithm}{Algorithms}

\crefname{ineq}{Inequality}{Inequalities}
\Crefname{ineq}{Inequality}{Inequalities}
\creflabelformat{ineq}{#2{\upshape#1}#3}

\crefname{prob}{Problem}{Problems}
\Crefname{prob}{Problem}{Problems}
\creflabelformat{prob}{#2{\upshape#1}#3} 

\crefname{ass}{Assumption}{Assumptions}
\Crefname{ass}{Assumption}{Assumptions}
\creflabelformat{ass}{#2{\upshape#1}#3} 

\crefname{ALC@line}{step}{steps}

\usetikzlibrary{patterns}
\newtheorem{theorem}{Theorem}

\newtheorem{definition}{Definition}

\newtheorem{lemma}{Lemma}

\newtheorem{corollary}{Corollary}

\newtheorem{remark}{Remark}

\DeclarePairedDelimiter\set{\{}{\}}

\newcommand{\sthat}{\text{s.t.}}

\newcommand{\col}{\text{col}}

\newcommand{\diag}{\text{diag}}
\newcommand{\tr}{\text{Trace}}

\newcommand{\he}{\text{He}}

\newcommand{\norm}[1]{\left\| #1 \right\|}

\newcommand{\brackets}[1]{\left( #1 \right)}

\newcommand{\quadsupplyinput}[5]{\begin{bmatrix}
    #4 \\ #5
\end{bmatrix}^{\hspace{-1pt}T}\hspace{-4pt}\begin{bmatrix}
    #1 & #2 \\ * & #3
\end{bmatrix}\hspace{-3pt}\begin{bmatrix}
    #4 \\ #5
\end{bmatrix}}

\newcommand{\dbx}{\dot{\textbf{x}}}

\newcommand{\cH}{\mathcal{H}}

\newcommand{\cL}{\mathcal{L}}

\newcommand{\cT}{\mathcal{T}}

\newcommand{\bbA}{\mathbb{A}}

\newcommand{\bbG}{\mathbb{G}}

\newcommand{\bbJ}{\mathbb{J}}
\newcommand{\bbK}{\mathbb{K}}

\newcommand{\bbN}{\mathbb{N}}

\newcommand{\bbQ}{\mathbb{Q}}
\newcommand{\bbR}{\mathbb{R}}
\newcommand{\bbS}{\mathbb{S}}

\newcommand{\bzero}{\mathbf{0}}

\newcommand{\bA}{\mathbf{A}}
\newcommand{\bB}{\mathbf{B}}
\newcommand{\bC}{\mathbf{C}}
\newcommand{\bD}{\mathbf{D}}

\newcommand{\bG}{\mathbf{G}}
\newcommand{\bH}{\mathbf{H}}
\newcommand{\bI}{\mathbf{I}}

\newcommand{\bK}{\mathbf{K}}
\newcommand{\bL}{\mathbf{L}}

\newcommand{\bN}{\mathbf{N}}

\newcommand{\bP}{\mathbf{P}}
\newcommand{\bQ}{\mathbf{Q}}
\newcommand{\bR}{\mathbf{R}}
\newcommand{\bS}{\mathbf{S}}
\newcommand{\bT}{\mathbf{T}}

\newcommand{\bW}{\mathbf{W}}
\newcommand{\bX}{\mathbf{X}}
\newcommand{\bY}{\mathbf{Y}}
\newcommand{\bZ}{\mathbf{Z}}

\newcommand{\be}{\mathbf{e}}

\newcommand{\bn}{\mathbf{n}}

\newcommand{\bu}{\mathbf{u}}

\newcommand{\bx}{\mathbf{x}}
\newcommand{\by}{\mathbf{y}}
\newcommand{\bz}{\mathbf{z}}

\newcommand{\scrG}{\mathscr{G}}

\newcommand{\scrK}{\mathscr{K}}

\newcommand{\barbH}{\overline{\bH}}

\newcommand{\barbQ}{\overline{\bQ}}
\newcommand{\barbR}{\overline{\bR}}
\newcommand{\barbS}{\overline{\bS}}

\newcommand{\barbX}{\overline{\bX}}

\newcommand{\barbZ}{\overline{\bZ}}

\newcommand{\hatbn}{\widehat{\bn}}

\newcommand{\hatbu}{\widehat{\bu}}

\newcommand{\hatby}{\widehat{\by}}

\newcommand{\hatbH}{\widehat{\bH}}

\newcommand{\hatbK}{\widehat{\bK}}

\newcommand{\hatbQ}{\widehat{\bQ}}
\newcommand{\hatbR}{\widehat{\bR}}
\newcommand{\hatbS}{\widehat{\bS}}
\newcommand{\hatbT}{\widehat{\bT}}

\newcommand{\hatbX}{\widehat{\bX}}

\newcommand{\hatbZ}{\widehat{\bZ}}

\newcommand{\tilbH}{\widetilde{\bH}}

\newcommand{\tilbT}{\widetilde{\bT}}

\newcommand{\tilbZ}{\widetilde{\bZ}}

\newcommand{\thh}{\mathrm{th}}

\newcommand{\barby}{\overline{\by}}

\newcommand{\barbn}{\overline{\bn}}
\newcommand{\barscrG}{\overline{\scrG}}

\newcommand{\dbX}{\delta\hspace{-.5pt}{\bX}}
\newcommand{\dbY}{\delta\hspace{-.5pt}{\bY}}
\newcommand{\dbK}{\delta\hspace{-.5pt}{\bK}}

\newcommand{\dbP}{\delta\hspace{-.5pt}{\bP}}

\newcommand{\dbA}{\delta\hspace{-.5pt}{\bA}}

\newcommand{\dbC}{\delta\hspace{-.5pt}{\bC}}

\newcommand{\dhatbK}{\delta\hspace{-.5pt}{\hatbK}}
\newcommand{\dhatbX}{\delta\hspace{-.5pt}{\hatbX}}

\newcommand{\dhatbQ}{\delta\hspace{-.5pt}{\hatbQ}}
\newcommand{\dhatbS}{\delta\hspace{-.5pt}{\hatbS}}
\newcommand{\dhatbR}{\delta\hspace{-.5pt}{\hatbR}}

\usepackage[acronym]{glossaries}

\makeglossaries
\newacronym{ndt}{NDT}{Network Dissipativity Theorem}
\newacronym{kyp}{KYP}{Kalman-Yakubovixh-Popov}
\newacronym{io}{IO}{input-output}
\newacronym{lti}{LTI}{linear time-invariant}
\newacronym[plural=LMIs,longplural=linear matrix inequalities]{lmi}{LMI}{linear matrix inequality}
\newacronym{lqr}{LQR}{linear quadratic regulator}
\newacronym{lqg}{LQG}{Linear Quadratic Gaussian}
\newacronym{lq}{LQ}{linear quadratic}

\newacronym{admm}{ADMM}{alternating direction methods of multipliers}
\newacronym{sdp}{SDP}{semidefinite programming}
\newacronym[plural=BMIs,longplural=bilinear matrix inequalities]{bmi}{BMI}{bilinear matrix inequality}
\newacronym{ico}{ICO}{iterative convex overbounding}

\newacronym{uav}{UAV}{unmanned aerial vehicle}
\newacronym{cps}{CPS}{cyber-physical system}

\title{\LARGE \bf
Consensus-Based Multi-Objective Controller Synthesis
}

\author{Ingyu Jang$^{1}$ and Leila Bridgeman$^{1}$% <-this % stops a space
\thanks{*This work is supported by ONR Grant No. N00014-23-1-2043.}% <-this % stops a space
\thanks{$^{1}$Ingyu Jang (PhD Student), and Leila Bridgeman (Assistant Professor) are with the Department of Mechanical Engineering and Material Science, Duke University, Durham, NC 27708 USA
        (email: {\tt\small ij40@duke.edu; ljb48@duke.edu}, phone: 919-225-4215)}%
}

\begin{document}

\maketitle
\thispagestyle{empty}
\pagestyle{empty}

%%%%%%%%%%%%%%%%%%%%%%%%%%%%%%%%%%%%%%%%%%%%%%%%%%%%%%%%%%%%%%%%%%%%%%%%%%%%%%%%
\begin{abstract}
Despite longstanding interest, controller synthesis remains challenging for networks of heterogeneous, nonlinear agents. 
Moreover, the requirements for computational scalability and information privacy have become increasingly critical.
This paper introduces a dissipativity-based distributed controller synthesis framework for networks with heterogeneous agents and diverse performance objectives, leveraging the Network Dissipativity Theorem and iterative convex overbounding.
Our approach enables the synthesis of controllers in a distributed way by achieving a network-wide consensus on agents' dissipativity variables while keeping sensitive subsystem information locally.
The proposed framework is applied to full-state feedback controller synthesis.

\end{abstract}

%%%%%%%%%%%%%%%%%%%%%%%%%%%%%%%%%%%%%%%%%%%%%%%%%%%%%%%%%%%%%%%%%%%%%%%%%%%%%%%%
\section{INTRODUCTION}
Classical analysis and control theory for networked systems were established in the 1970s \cite{moylan1978stability,vidyasagar1981input}, but modern multi-agent systems, such as power grids and swarm robotics, exhibit structural complexities that exceed the capabilities of these traditional frameworks \cite{sztipanovits2011toward}.
These challenges have motivated research on compositional stability analysis and structured controller synthesis \cite{caliskan2014compositional,jovanovic2016controller,lian2018sparsity}. 
Nevertheless, most existing synthesis methods rely on centralized computation or global information exchange during the offline design phase, which incurs significant scalability bottlenecks and concerns regarding data privacy or intellectual property.
To address these limitations, this work integrates energy-based compositional controller synthesis with a distributed optimization framework.
Our approach enables the design of a global controller through purely local computation to achieve a network-wide consensus on common dissipativity parameters, significantly reducing the communication overhead and eliminating the need to share sensitive subsystem parameters across the network.

Dissipativity \cite{willems1972dissipative,hill1977stability} plays a fundamental role in the stability analysis of physical systems by leveraging \gls{io} relationships rather than internal states.
Crucially for multi-agent systems, the \gls{ndt} established that interconnected dissipative subsystems yield dissipativity of the network and, under additional conditions, stability \cite{moylan1978stability,vidyasagar1981input}. This modularizes stability analysis by using only open-loop dissipativity characteristics of individual agents. 
% Its inherent compositionality is particulary advantageous for multi-agent systems, as interconnected dissipative subsystems yield dissipativity of the network and, under additional conditions, stability.
% In \cite{moylan1978stability,vidyasagar1981input} extend these advantages by introducing \gls{ndt}, which modularizes the stability analysis by using only open-loop dissipativity characteristics of individual agents.
This decouples local dynamics from the network topology, thereby facilitating the analysis and control of networked systems with heterogeneous and nonlinear agents \cite{arcak2016networks,locicero2025dissipativity}.
These modular properties make \gls{ndt} uniquely compatible with distributed optimization \cite{yang2019survey}, enabling an approach that eliminates the need for global information exchange during the design phase.

% Dissipativity \cite{willems1972dissipative,hill1977stability} plays a fundamental role in the stability analysis of physical systems by leveraging \gls{io} relationships rather than internal states.
% Its inherent compositionality is particulary advantageous for multi-agent systems, as interconnected dissipative subsystems yield dissipativity of the network and, under additional conditions, stability.
% The work in \cite{moylan1978stability,vidyasagar1981input} extend these advantages by introducing \gls{ndt}, which modularizes the stability analysis by using only open-loop dissipativity characteristics of individual agents.
% This modularization decouples local dynamics from the network topology, thereby facilitates the analysis and control of networked systems with heterogeneous and nonlinear agents \cite{arcak2016networks,locicero2025dissipativity}.
% These modular properties make \gls{ndt} uniquely compatible with distributed optimization \cite{yang2019survey}, enabling an approach that eliminates the need for global information exchange during the design phase.

Growing demands for security and computational efficiency have driven interest in distributed optimization for networked systems \cite{nedich2015convergence,nedic2018network,yang2019survey}.
Recent efforts have integrrated distributed optimization with dissipativity concepts.
In \cite{agarwal2020distributed}, a dissipativity-based distributed synthesis using an iterative Schur complement was proposed, but the approach is restricted to \gls{lti} agents and lacks a formal performance objective.
For stability certification, \cite{meissen2015compositional,jang2025consensus} used the \gls{ndt} and distributed optimization to decompose network stability analysis. 

This paper adapts the analysis strategies in \cite{meissen2015compositional,jang2025consensus} to synthesis by enforcing consensus on the variables associated with a global constraint using \gls{admm} and \gls{ico}.
The proposed framework adopts abstract local controllers, admitting flexible controller implementations, including full-state feedback, observer-based controllers, or neural network controllers.
Furthermore, by leveraging distributed optimization, our approach handles heterogeneous agents with diverse performance objectives.
A case study with full-state feedback controllers is presented to illustrate the practical implementation of the proposed framework.

\section{PRELIMINARIES} \label{chap:Preliminaries}

\subsection{Notation}
The sets of real numbers and natural numbers up to $n$ are denoted by $\bbR$ and $\bbN_n$, respectively. 
The set of real $n{\times}m$ matrices and $n{\times}n$ symmetric matrices are $\bbR^{n\times m}$ and $\bbS^n$, respectively. 
For $\bA{\in}\bbR^{\sum_{i=1}^Nn_i{\times}\sum_{j=1}^N m_j}$, $(\bA)_{ij}{\in}\bbR^{n_i{\times}m_j}$ is its $(i,j)^\thh$ block, then $\bA$ is said to be a block-wise matrix in $\bbR^{N\times M}$.
The block diagonal matrix of $\bA_i$ for all $i{\in}\bbN_N$ is $\diag(\bA_i)_{i=1}^N$. 
$\col(\bx_i)_{i=1}^N$ represents $[\bx_1^T{,}{\dots}{,}\bx_N^T]^T$ for vectors $\bx_i{\in}\bbR^{n_i}$.
The notation $\bA{\prec}0$ and $\bA{\preceq}0$ indicate that $\bA$ is negative-definite and negative semi-definite, respectively. 
For brevity, $\he(\bA){=}\bA{+}\bA^T$ and asterisks, $*$, denote duplicate blocks in symmetric matrices.
$\cT_0^1(\bA)$ is the 1st order Taylor expansion of the matrix variable $\bA$ from its initial point $\bA^0$, meaning $\cT_0^1(\bA){=}\bA^0{+}\delta\bA$.
The $n{\times}n$ identity matrix and zero matrix are denoted by $\bI_n$ and $\bzero_{n\times n}$, respectively. 

The set of square integrable functions is $\cL_{2}$. 
The Frobenius norm and $\cL_2$ norm are denoted by $\|{\cdot}\|_F$ and $\|{\cdot}\|_2$, respectively. 
The truncation of a function $\by(t)$ at $T$ is denoted by $\by_T(t)$, where $\by_T(t){=}\by(t)$ if $t\leq T$, and $\by_T(t){=}0$ otherwise. 
If $\|\by_T\|_2^2{=}{\langle}\by_T,\by_T{\rangle}{=}\int_0^{\infty}\by_T^T(t)\by_T(t)dt{<}\infty$ for all $T{\geq}0$, then $\by{\in}\cL_{2e}$, where $\cL_{2e}$ is the extended $\cL_2$ space. 
The indicator function is denoted by $I_\bbA{:}\Omega{\to}\set{0,1}$, where $I_\bbA(x){=}0$ if $x{\in}\bbA$ and $I_\bbA(x){=}1$ otherwise for all $x{\in}\Omega$.

An $N$-tuple $(\bA_i)_{i=1}^N{\in}\bbR^{m_1{\times}n_1}{\times}{\cdots}{\times}\bbR^{m_N{\times}n_N}$ is composed of $\bA_i{\in}\bbR^{m_i{\times}n_i}$ for $i{\in}\bbN_N$, and is equivalent to $(\bA_1{,}{\dots}{,}\bA_N)$.
Given two $N$-tuples $\bA$ and $\bB$, $\bA{+}\bB$ denotes the component-wise sum $(\bA_i{+}\bB_i)_{i=1}^N$.
Other algebraic operations are defined analogously.
For given $N$-tuple $\bA$, its norm is defined as the sum of the constituent norms, i.e., $\norm{\bA}{=}\sum_{i=1}^N\norm{\bA_i}$.

\subsection{Alternating Direction Method of Multipliers (ADMM)} \label{chap:ADMM}

\gls{admm} is a powerful algorithm for distributed optimization\cite{boyd2011distributed}. Consider the following constrained problem,
\begin{align} \label{eq:Constrained Optimization Problem}
\begin{split}
    \min_{\substack{
            \scriptscriptstyle \bX_i,\forall i\in\bbN_N
        }}\;\;&\sum_{i\in\bbN_N}f_i(\bX), \quad
    \sthat\;\;\bX_i\in\Omega_i\;\forall i\in\bbN_N,\;\bX\in\Psi,
\end{split}
\end{align}
where $\bX_i{\in}\bbR^{n_i\times m_i}$, $f_i{:}\bbR^{n_1\times m_1}{\times}{\cdots}{\times}\bbR^{n_N\times m_N}{\to} \bbR{\cup}\set{{+}\infty}$, and $\Omega_i$ are $i^\thh$ local variable, objective function, and constraint set defined on $\bbR^{n_i{\times} m_i}$, respectively, $\bX{=}(\bX_i)_{i\in\bbN_N}$, and $\Psi$ is the global constraint set. 
By introducing $N$ slack variables, $\bY_i{=}\bX$ for all $i{\in}\bbN_N$, \cref{eq:Constrained Optimization Problem} can be reformulated as
\begin{align} \label{eq:Equivalent Constrained Optimization Problem}
    \begin{split}
        \min_{\substack{
            \scriptscriptstyle \bY_i,\forall i\in\bbN_N,\bZ
        }}&\sum_{i\in\bbN_N}f_i(\bY_i)+\sum_{i\in\bbN_N}I_{\Omega_i^{\bY}}(\bY_i)+I_\Psi(\bZ) \\[-3pt]
        \sthat\quad&\bY_i-\bZ=\bzero\quad\forall i\in\bbN_N,
    \end{split}
\end{align}
with global ``consensus" variable $\bZ$, where $\Omega_i^{\bY}$ is the same set as $\Omega_i$ but defined over $\bbR^{n_1\times m_1}{\times}{\cdots}{\times}\bbR^{n_N\times m_N}$. \gls{admm} iteratively solves \cref{eq:Equivalent Constrained Optimization Problem} through
\begin{subequations} \label{eq:General ADMM}
    \begin{align}
        \bY_i^{k+1}
            &{=}\arg\min_{\substack{
                    \scriptscriptstyle \bY_i
                }}\Big(f_i(\bY_i){+}I_{\Omega_i^\bY}(\bY_i)
                {+}\frac{\rho}{2}\|\bY_i
                    {-}\bZ^k
                    {+}\bT_i^k\|_F^2\Big), 
                    \label{eq:General ADMM x Update} \\[-3pt]
        \bZ^{k+1}
            &{=}\arg\min_{\substack{
                    \scriptscriptstyle \bZ
                }}\Big(I_\Psi(\bZ)
                {+}\sum_{i\in\bbN_N}\frac{\rho}{2}\|\bY_i^{k+1}
                    {-}\bZ
                    {+}\bT_i^k\|_F^2\Big) 
            %         \nonumber \\ 
            % &{=} \Pi_\Omega(\bX^{k+1}+\bT^k), 
                    \label{eq:General ADMM z Update} \\[-3pt]
        \bT_i^{k+1}
            &=\bT_i^k+(\bY_i^{k+1}
                -\bZ^{k+1}), 
                \label{eq:General ADMM u Update}
    \end{align}
\end{subequations}
where $\bT_i{\in}\bbR^{n_1{\times} m_1}{\times}{\cdots}{\times}\bbR^{n_N{\times} m_N}$ is the $i^\thh$ ``dual" variable corresponding to $\bY_i$, $k$ is the iteration index, and $\rho{>}0$ is a hyper parameter \cite{boyd2011distributed}. 
The stopping criteria of \gls{admm} are $r_p^i{=}\frac{\|\bY_i^k{-}\bZ^k\|_F}{\|\bY_i\|_F}{\leq}\epsilon_p$ for all $i{\in}\bbN_N$ and $r_d{=}\frac{\|\bZ^k{-}\bZ^{k{-}1}\|_F}{\|\bZ\|_F}{\leq}\epsilon_d$
If $f_i$ and $I_{\Omega_i^\bY}$ for all $i{\in}\bbN_N$, and $I_\Psi$ are closed, proper, and convex, \cref{eq:Equivalent Constrained Optimization Problem} is a convex optimization with equality constraints, which satisfies Slater's condition, implying that its Lagrangian has a saddle point \cite{boyd2004convex}. If \cref{eq:Equivalent Constrained Optimization Problem} is feasible and its Lagrangian has a saddle point, \gls{admm} guarantees that as $k{\to}\infty$: $\bY_i^{k}{-}\bZ^k{\to}\bzero$; $\sum_{i=1}^Nf_i(\bY_i^k){+}I_{\Omega_i^\bY}(\bY_i^k){+}I_\Psi(\bZ){\to}\sum_{i=1}^Nf_i(\bY^\star)$; and $\bT_i^k{\to}\bT_i^\star$ for all $i{\in}\bbN_N$, where $\bY^\star{=}\bX^\star$ is a solution of \cref{eq:Constrained Optimization Problem} \cite{boyd2011distributed}.

\subsection{\texorpdfstring{$\bQ\bS\bR$}{QSR}-Dissipativity of Large-Scale Multi-Agent Systems}
$\bQ\bS\bR$-dissipativity relates system inputs and outputs.
\begin{definition} [$\bQ\bS\bR$-Dissipativity \cite{vidyasagar1981input}] \label{def:QSR}
    Let $\bQ{\in}\bbS^{l}$, $\bR{\in}\bbS^{m}$, $\bS{\in}\bbR^{l\times m}$. The system $\scrG{:}\cL_{2e}^m{\to}\cL_{2e}^l$ is \textit{$\bQ\bS\bR$-dissipative} if there exists $\beta{\in}\bbR$ such that for all $\bu{\in}\cL_2^m$ and $T{>}0$,
    \begin{align} \label{eq:QSR Dissipativity}
        \int_0^T\quadsupplyinput{\bQ}{\bS}{\bR}{\scrG(\bu(t))}{\bu(t)}dt\geq\beta.
    \end{align}
\end{definition}
\cref{lem:KYP Lemma} establishes $\bQ\bS\bR$-dissipativity of \gls{lti} systems.
\begin{lemma}[\cite{gupta1996robust}] \label{lem:KYP Lemma} 
    An \gls{lti} system with minimal realization $\Sigma{:}\dot{\bx}{=}\bA\bx{+}\bB\bu,\ \by{=}\bC\bx{+}\bD\bu$ is $\bQ\bS\bR$-dissipative if there exist matrices $\bP{\succ}0$, $\bQ$, $\bS$, and $\bR$ satisfying
    \setlength{\arraycolsep}{2.5pt}
    \begin{align} \label{eq:KYP Lemma}
        \begin{bmatrix}
            \bA^T\bP{+}\bP\bA{-}\bC^T\bQ\bC & \bP\bB{-}\bC^T\bS{-}\bC^T\bQ\bD \\
            \bB^T\bP{-}\bS^T\bC{-}\bD^T\bQ\bC & {-}\bR{-}\bS^T\bD{-}\bD^T\bS{-}\bD^T\bQ\bD
        \end{bmatrix}\preceq0.
    \end{align}
\end{lemma}
$\bQ\bS\bR$-dissiaptivity is useful for ensuring $\cL_2$-stability. 
\begin{definition}[$\cL_2$-stability \cite{vidyasagar1981input}] \label{def:L2 Stable}
    An operator $\scrG{:}\cL_{2e}^m{\mapsto}\cL_{2e}^n$ is $\cL_2$-stable if for any $\bu{\in}\cL_2^m$ and all $\bx_0$ there exist a constant $\gamma{>}0$ and a function $\beta(\bx_0)$ such that
    \begin{align}
        \|(\scrG\bu)_T\|_2\leq\gamma\|\bu_T\|_2+\beta(\bx_0),\quad T>0.
    \end{align}
\end{definition}

For multi-agent systems, \gls{ndt}, stated below, relates the dissipativiy of each agent to $\cL_2$ stability of entire system.

\begin{theorem}[\gls{ndt} \cite{vidyasagar1981input}] \label{thm:ndt}
    Consider a multi-agent system, $\scrG{:}\bu{\to}\by$, composed of $N$ $\bQ_i\bS_i\bR_i$-dissipative agents, $\scrG_i:\cL^{m_i}_{2e}{\to}\cL^{l_i}_{2e}$, with mappings $\by_i{=}\scrG_i\bu_i$ interconnected as
    \begin{align}\label{eq:Interconnected system}
        \begin{split}
            \by=\scrG\be, \qquad \bu=\be+\barbH\by, 
        \end{split}
    \end{align}
    where $\bu{=}\col(\bu_i)_{i{=}1}^N$, $\by{=}\col(\by)_{i{=}1}^N$, $\be{=}\col(\bu_i)_{i{=}1}^N$ are the aggregated input of each agent, output of each agent, and the exogenous input to the multi-agent network, respectively,, and $\barbH$ is the interconnection matrix with $(\barbH)_{ii}{=}\bzero$.
    Then $\mathscr{G}$ is $\cL_2$ stable if %$\barbQ\prec 0$, where
    \begin{align}\label{eq:ndt}
        \barbQ+\barbS\barbH+\barbH^T\barbS^T+\barbH^T\barbR\barbH\prec0,
    \end{align}
    with $\barbQ{=}\diag(\barbQ_i)_{i=1}^N$, $\barbR{=}\diag(\bR_i)_{i=1}^N$, and $\barbS{=}\diag(\bS_i)_{i=1}^N$.
\end{theorem}

\subsection{Iterative Convex Overbounding (ICO)} \label{subChap:convex_overbounding}
Optimal control synthesis problems frequently involve nonconvex \glspl{bmi} of the form
\begin{align} \label{eq:bmi}
    \bQ+\he(\bX\bN\bY)\prec0,
\end{align}
where $\bN{\in}\bbR^{p{\times}q}$ is fixed, and $\bQ{\in}\bbS^n$, $\bX{\in}\bbR^{n{\times}p}$. and $\bY{\in}\bbR^{q{\times}n}$ are design variables.
\cref{thm:overbounding} averts the NP-hardness of \cref{eq:bmi}, by solving a restricted, but convex problem.
\begin{theorem} [\cite{sebe2018sequential}] \label{thm:overbounding}
    Consider the matrices $\bQ{\in}\bbS^n$, $\bN\in\bbR^{p\times q}$, $\bX{\in}\bbR^{n\times p}$. and $\bY{\in}\bbR^{q\times n}$, where $\bQ$, $\bX$, and $\bY$ are design variables.
    The \gls{bmi} condition $\bQ{+}\he(\bX\bN\bY){\prec}0$ is implied
    \begin{align} \label{eq:overbounding_sebe}
        \begin{bmatrix}
            \bQ & \bX\bN{+}\bY^T\bG^T \\
            \bN^T\bX^T{+}\bG\bY & -\he(\bG)
        \end{bmatrix}{\prec}0
    \end{align}
    for any $\bG{\in}\bbR^{q{\times}q}$ satisfying $\he(\bG){\succ}0$.
\end{theorem}

The conservatism introduced by \cref{eq:overbounding_sebe} can be mitigated by iteratively updating a feasible point, $(\bX^i,\bY^i)$, satisfying \cref{eq:bmi}.
By defining $(\bX^{i+1},\bY^{i+1})=(\bX^i{+}\dbX,\bY^i{+}\dbY)$, $\dbX$ and $\dbY$ serve as the decision variables. 
The tightening of \cref{eq:overbounding_sebe} relative to \cref{eq:bmi} is proportional to these perturbations. 
As detailed in \cite{warner2017iterative}, this iterative scheme reduces the inherent conservatism of \cref{thm:overbounding}.
Each optimization problem remains feasible since $\dbX{=}\bzero$ and $\dbY{=}\bzero$ yield the initial feasible point.

\begin{remark}
    In this paper, $\bI$ is used as $\bG$, but any $\bG$ satisfying $\he(\bG){\succ}0$ can be used for \cref{eq:overbounding_sebe}. 
\end{remark}

\section{Distributed Controller Synthesis} \label{chap:Distributed Controller Synthesis}

\subsection{Multi-Agent Networked System}
Consider a multi-agent networked system $\bz{=}\scrG(\be)$, where $\be{\in}\bbR^{\sum_{i=1}^Nm_i}$ and $\bz{\in}\bbR^{\sum_{i=1}^Nl_i}$ are the exogenous input and output of $\scrG$, composed of $N$ heterogeneous agents $\by_i{=}\scrG_i(\bu_i)$, where $\bu_i{\in}\bbR^{m_i}$ and $\by_i{\in}\bbR^{l_i}$ are the local input and output of $\scrG_i$.
The interconnection topology between the agents and the exogenous network signals is described by
\begin{align}
    \bu=\tilbH\be+\bH\by,\quad
    \bz=\hatbH\by
\end{align}
where $\bu{=}\col(\bu_i)_{i=1}^N$, $\by$, $\be$, and $\bz$ are defined analogously.
$\bH$, $\tilbH$, and $\hatbH$ characterize the inter-agent couplings, the mapping from the network's exogenous inputs to agent inputs, and the mapping from agent outputs to network outputs, respectively.
Assume no self-feedback loops, $(\bH)_{ii}{=}\bzero$ for all $i{\in}\bbN_N$.

% We introduce a feedback control layer to achieve the specific performance of the entire network.
% Consider local controllers $\hatby_i{=}\scrK_i(\hatbu_i)$, where $\hatbu_i{\in}\bbR^{l_i}$ and $\hatby_i{\in}\bbR^{m_i}$ are the input and output of $i^\thh$ local controller.
% These local units constitute a global control law $\hatby{=}\scrK(\hatbu)$ with output $\hatby$ and input $\hatbu{=}\bz{+}\hatbn$, where $\hatby{=}\col(\hatby_i)_{i=1}^N$ and $\hatbu{=}\col(\hatbu_i)_{i=1}^N$. $\hatbn\in\bbR^{\sum_{i=1}^Nl_i}$ is the exogenous input to this global controller, such as measurement noise.

We introduce a feedback control layer to achieve specific agent-wide performance.
Consider local controllers $\hatby_i{=}\scrK_i(\hatbu_i)$, where $\hatbu_i{\in}\bbR^{l_i}$ and $\hatby_i{\in}\bbR^{m_i}$ are the input and output of $\scrK_i$.
These decentralized local units constitute a global control law $\hatby{=}\scrK(\hatbu)$, where $\hatby{=}\col(\hatby_i)_{i=1}^N$ and $\hatbu{=}\col(\hatbu_i)_{i=1}^N$. 
% $\hatbn\in\bbR^{\sum_{i=1}^Nl_i}$ is the exogenous input to this global controller, such as measurement noise.

The feedback interconnection between each agent and its local controller forms the closed-loop networked system $\barby{=}\barscrG(\barbn)$ as illustrated in \cref{fig1a:global}, where $\barbn{=}\col(\bn,\hatbn)$ and $\barby{=}\col(\by,\hatby)$.
Here, $\bn$ indicates the exogenous disturbances to $\scrG$ and $\hatbn$ indicates the measurement noise inputs to $\scrK$.
Consequently, the augmented interconnection between agents, controllers, and the network is expressed as
\begin{align} \label{eq:io_global_network}
    \begin{bmatrix}
        \bu \\ \hatbu
    \end{bmatrix}
    {=}\begin{bmatrix}
        \bH\by{+}\tilbH\hatby{+}\tilbH\bn \\
        \hatbH\by{+}\hatbn
    \end{bmatrix}
    {=}\barbH
    \begin{bmatrix}
        \by \\ \hatby
    \end{bmatrix}
    {+}\begin{bmatrix}
        \tilbH\bn \\ \hatbn
    \end{bmatrix},
    \barbH{=}\begin{bmatrix}\;
        \bH & \tilbH \\ \hatbH & \bzero
    \end{bmatrix}
\end{align}
as depicted in \cref{fig1b:module}. 
According to \gls{ndt}, if each $\scrG_i$ is $\bQ_i\bS_i\bR_i$-dissipative and each $\scrK_i$ is $\hatbQ_i\hatbS_i\hatbR_i$-dissipative, $\barscrG$ is $\cL_2$ stable if it satisfies \cref{eq:ndt} with $\barbH$, $\barbQ=\diag(\diag(\bQ_i)_{i=1}^N,\diag(\hatbQ_i)_{i=1}^N)$, $\barbS=\diag($ $\diag(\bS_i)_{i=1}^N,$ $\diag(\hatbS_i)_{i=1}^N)$, and
$\barbR{=}\diag(\diag(\bR_i)_{i=1}^N{,}\diag(\hatbR_i)_{i=1}^N)$.
% \begin{align*}
%     \barbQ&=\diag(\diag(\bQ_i)_{i=1}^N,\diag(\hatbQ_i)_{i=1}^N), \\
%     \barbS&=\diag(\diag(\bS_i)_{i=1}^N,\diag(\hatbS_i)_{i=1}^N), \\
%     \barbR&=\diag(\diag(\bR_i)_{i=1}^N,\diag(\hatbR_i)_{i=1}^N).
% \end{align*}

\begin{figure}
    \captionsetup[sub]{aboveskip=1pt, belowskip=1pt}
    \centering
    \subfloat[Network representation]{%
    \includegraphics[height=0.083\textheight]{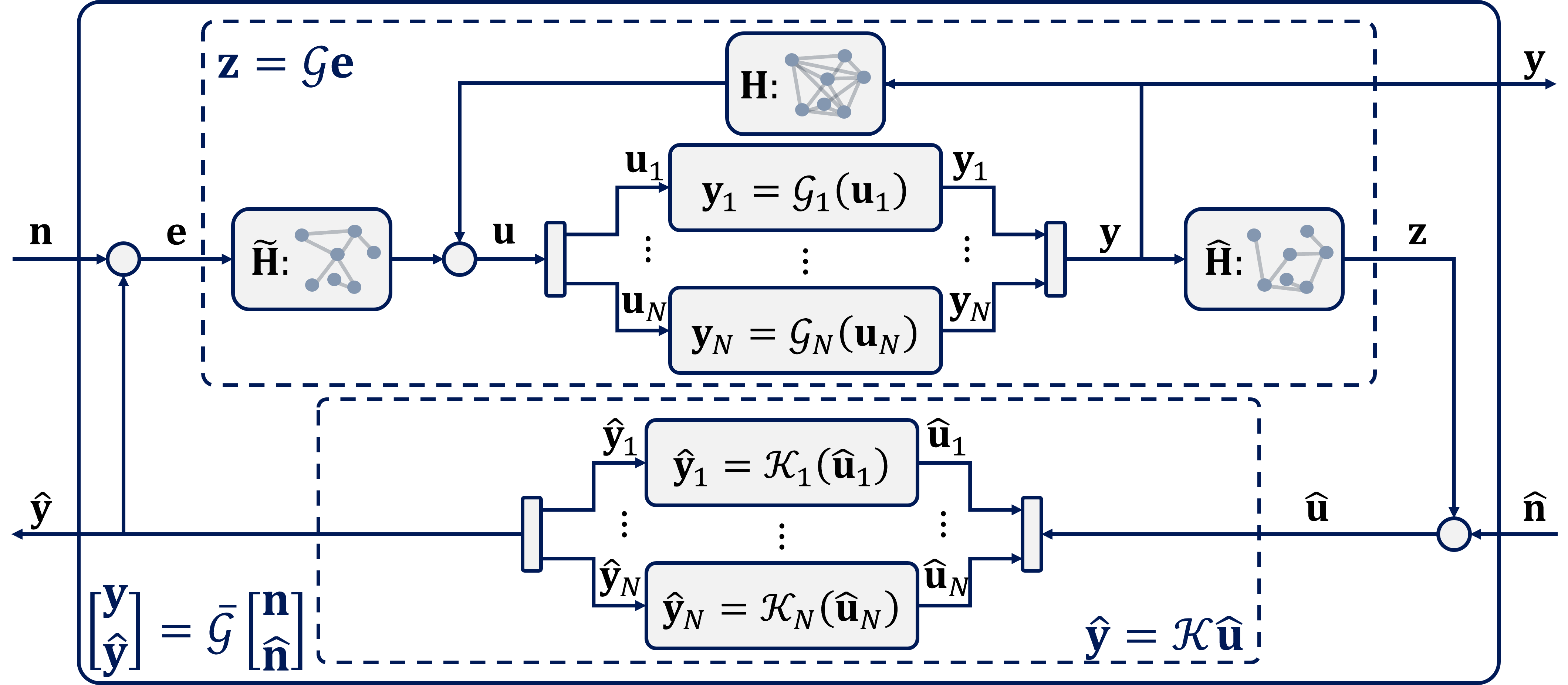} \label{fig1a:global}
    }
    \subfloat[Modular representation]{%
    \includegraphics[height=0.083\textheight]{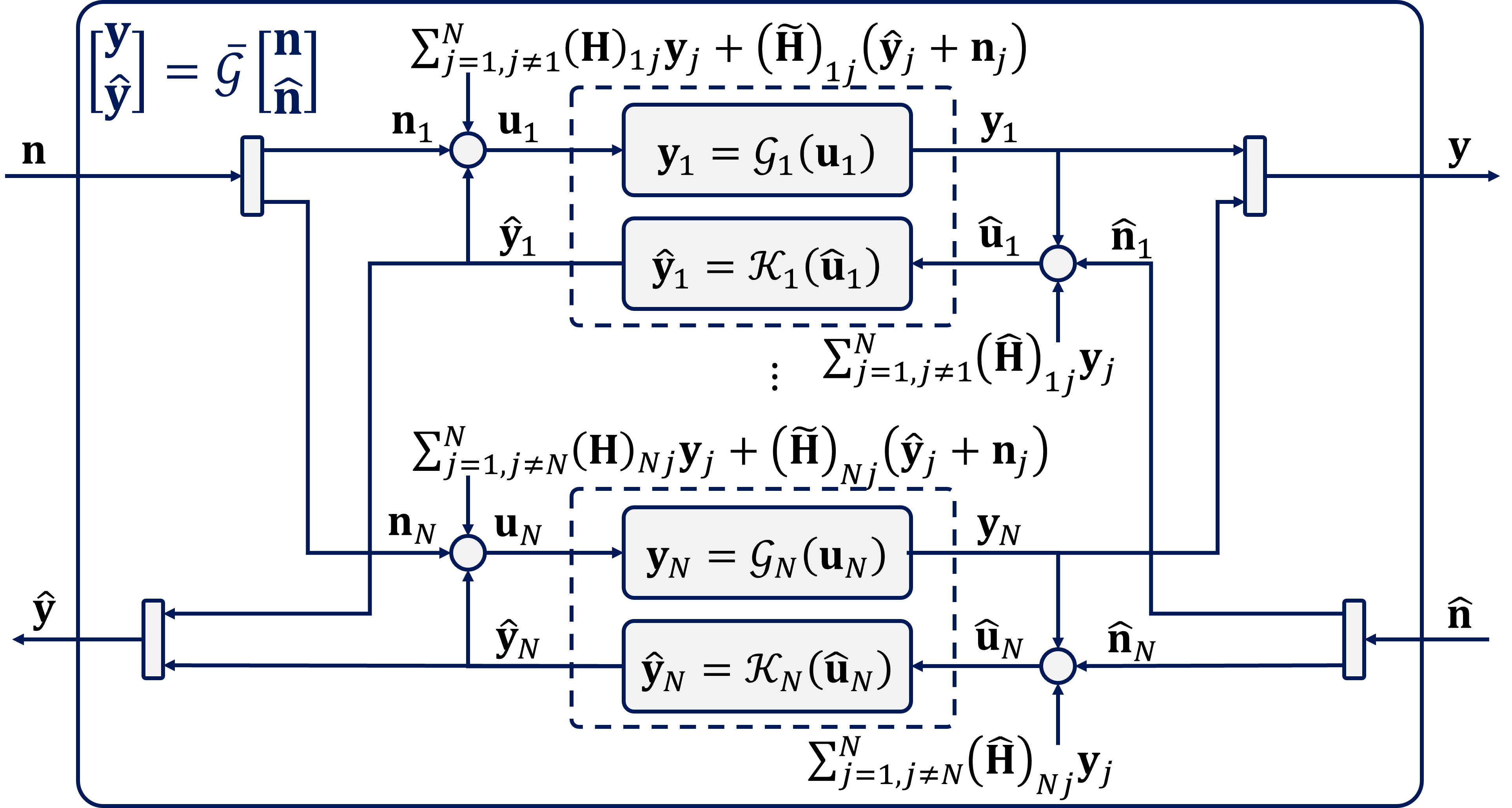} \label{fig1b:module}
    }
    \vspace{-1pt}
    \caption{Two representations of a multi-agent system.}
    \label{fig1:network_system}
    \vspace*{-1.25\baselineskip} 
\end{figure}

% interconnected through $\bH$. The overall network dynamics are described by
% \begin{align} \label{eq:plant}
%     {\setlength{\arraycolsep}{4pt}\begin{array}{rlll}
%         \scrG_i: & \dbx_i=f_i(\bx_i,\bu_i), & \by_i=h_i(\bx_i) \\
%         \scrG: & \dbx=f(\bx,\be),&\bz=h(\bx), & \bu=\be+\bH\by,\;\;\bz=\by
%     \end{array}}
% \end{align}
% where $\bx_i\in\cL_{2e}^{n_i}$, $\bu_i\in\cL_{2e}^{m_i}$, and $\by_i\in\cL_{2e}^{l_i}$ are the states, inputs, and outputs of the $i^\thh$ agent, respectively. $\bu$, $\be$, and $\by$ are stacked vectors defined in \cref{thm:ndt}, where $\be$ represents the exogenous input to the global network. The global network output is denoted by $\bz$, which is equivalent to $\by$.

\subsection{Controller Synthesis}
Consider the performance objective $J_i(\hatbK_i)$ of the $i^\thh$ closed-loop formed by $\scrG_i$ and $\scrK_i$, where $\hatbK_i{=}(\bK_i)_{i=1}^N$ is the tuple of design parameters defining $J_i$, and $\bK_i$ parameterizes $\scrK_i$. 
The optimal performance of $\barscrG$, subject to the $\cL_2$ stability, is formulated as
\begin{subequations} \label{eq:Main Problem}
    \begin{align}
        \min_{\substack{
            \scriptscriptstyle \bK_i,\bX_i,
        }}\quad&\sum_{i\in\bbN_N}J_i(\hatbK_i),&\hspace{-40pt}\forall i\in\bbN_N,\label{eq:Main:objective} \\[-3pt]
        \sthat\quad&\hatbK_i\in\bbJ_i,&\hspace{-40pt}\forall i\in\bbN_N, \label{eq:Main:lyapunov}  \\[-3pt]
            &\bX_i\in\bbG_i,&\hspace{-40pt}\forall i\in\bbN_N, \label{eq:Main:plant} \\[-3pt]
            &(\hatbK_i,\hatbX_i)\in\bbK_i,&\hspace{-40pt}\forall i\in\bbN_N, \label{eq:Main:ctrl} \\[-3pt]
            &(\bX,\hatbX)\in\bbQ,&\hspace{-40pt} \label{eq:Main:stability} 
    \end{align}
\end{subequations}
where $\bX_i{=}(\bQ_i,\bS_i,\bR_i)$ and $\hatbX_i{=}(\hatbQ_i,\hatbS_i,\hatbR_i)$ are 3-tuples denoting dissipativity parameters of $\scrG_i$ and $\scrK_i$, respectively.
$\bX{=}(\bX_i)_{i=1}^N$ and $\hatbX{=}(\hatbX_i)_{i=1}^N$, which are the network-wide tuples necessary for global stability.
$\bbJ_i{=}\{\hatbK_i\,|\,$elements in $\hatbK_i$ satisfy constraints defining $J_i\}$,
$\bbG_i{=}\{\bX_i\,|\,$elements in $\bX_i$ satisfy \cref{eq:QSR Dissipativity}$\}$,
$\bbK_i{=}\{(\hatbK_i{,}\hatbX_i)\,|\,$elements in $\hatbK_i$ and $\hatbX_i$ satisfy \cref{eq:QSR Dissipativity}$\}$, and
$\bbQ{=}\{(\bX,\hatbX)\,|\,$\cref{thm:ndt} holds$\}$.

Typical approaches to solving \cref{eq:Main Problem} consist of first determining $\bX_i$ satisfying \cref{eq:Main:plant}, then finding $\hatbX_i$ satisfying \cref{eq:Main:stability}, and finally solving \cref{eq:Main Problem} with $\bX_i$ and $\hatbX_i$ fixed.
However, this conservatively ignores the infinite variety of dissipativity parameter combinations that could potentially satisfy the stability requirement.
Therefore, co-optimizing parameters in \cref{eq:Main Problem} is essential, rather than fixing them a priori.

Although identifying dissipativity using \cref{eq:Main:plant} may appear nonconvex, various formulations of \cref{def:QSR} provide computationally tractable dissipativity definitions. 
For example, when $\scrG_i$ is a \gls{lti} system, a \gls{lti} system with input/state/output delays, or an input-affine nonlinear system, the \glspl{lmi} in \cref{lem:KYP Lemma}, \cite[Theorem 3.1]{bridgeman2016conic}, or \cite[Theorem 4]{Strong2024IterativeGain} gives systematic frameworks that can be solved via standard \gls{sdp} solvers.
If the dissipativity parameters are known a priori, they can be scaled by a scalar design variable $\lambda_i{\geq}0$ to enhance the flexibility of the synthesis. For example, systems governed by Euler-Lagrange dynamics are inherently dissipative with $(\bzero,\bI,\bzero)$, yielding $(\bQ_i,\bS_i,\bR_i){=}(\bzero,\lambda_i\bI,\bzero)$ for \cref{eq:Main:plant}.

Despite these advantages, solving \cref{eq:Main Problem} directly remains challenging. 
First, centralized computation requires agents to share their internal state-space dynamics, raising privacy concerns.
Second, \cref{eq:Main:ctrl} remains nonconvex due to the coupling of $\hatbK_i$ and $\hatbX_i$, which results in high-order matrix inequalities, rendering the problem NP-hard.
In addition, most constraints related to \cref{eq:Main:lyapunov} introduce additional high-order matrix inequalities, such as those arising in $\cH_2$ or $\cH_\infty$ norm objectives.
This paper overcomes these limitations by integrating distributed optimization with \gls{ico}.

\subsection{Distributed Controller Synthesis}
This section develops an algorithm to solve \cref{eq:Main Problem} in a distributed manner.
Under this framework, agents preserve their privacy by solving local optimal control problems independently, without disclosing their internal dynamics. 
The resulting dissipativity parameters $\bX_i$ and $\hatbX_i$ are shared and iteratively adjusted to satisfy \gls{ndt}.

By setting $\Omega_i{=}\bbJ_i{\cap}\bbG_i{\cap}\bbK_i$ for all $i{\in}\bbN_N$ and $\Psi{=}\bbQ$, \cref{eq:Main Problem} fits the formation in \cref{eq:Constrained Optimization Problem}.
Therefore, \gls{admm} can be applied as
\begin{subequations} \label{eq:admm_controller_synthesis}
    \begin{align}
        \hspace{-5pt}\barbX_i^{k+1}
            &\hspace{-4pt}{=}\arg\!\!\min_{\substack{
                    \scriptscriptstyle \!\!\!\!\!\!\!\bX_i,\hatbX_i,\hatbK_i
                }}J_i(\hatbK_i)
                {+}\frac{\rho}{2}\hspace{-2pt}\brackets{
                \hspace{-6pt}\begin{array}{l}
                    \|\bX_i
                    {-}\bZ_i^k
                    {+}\bT_i^k\|_F^2
                    {+}\|\hatbX_i   
                    {-}\hatbZ_i^k
                    {+}\hatbT_i^k\|_F^2 \\
                    \quad{+}\|\hatbK_i{-}\tilbZ^k{+}\tilbT_i^k\|_F^2
                \end{array}\hspace{-6pt}
                }\hspace{-2pt}{,}  \nonumber\\[-4pt]
            &\hspace{-2pt}\quad\quad\sthat\quad
            \hatbK_i{\in}\bbJ_i,\;\;
            \bX_i{\in}\bbG_i,\;\;
            (\hatbK_i,\hatbX_i){\in}\bbK_i, \label{eq:admm_x Update controller synthesis} \\
        \hspace{-5pt}\barbZ^{k+1}
            &\hspace{-4pt}{=}\arg\!\!\min_{\substack{
                    \scriptscriptstyle \!\!\!\!\!\!\!\bZ,\hatbZ,\tilbZ
                }}\sum_{i\in\bbN_N}\brackets{
                \hspace{-6pt}\begin{array}{l}
                    \|\bX_i^{k+1}
                    {-}\bZ_i
                    {+}\bT_i^k\|_F^2
                    {+}\|\hatbX_i^{k+1}   
                    {-}\hatbZ_i^k
                    {+}\hatbT_i^k\|_F^2 \\
                    \quad{+}\|\hatbK_i^{k+1}{-}\tilbZ{+}\tilbT_i^k\|_F^2
                \end{array}\hspace{-6pt}
                }\hspace{-2pt}, \nonumber\\[-4pt]
            &\hspace{-2pt}\quad\quad\sthat\quad
            (\bZ,\hatbZ)\in\bbQ, \label{eq:admm_z Update controller synthesis} \\
            \begin{split}
                \hspace{-5pt}\bT_i^{k+1}
                    &\hspace{-4pt}{=}\bT_i^k{+}(\bX_i^{k+1}
                        {-}\bZ_i^{k+1}),\;\;
                \hatbT_i^{k+1}
                    {=}\hatbT_i^k{+}(\hatbX_i^{k+1}
                        {-}\hatbZ_i^{k+1}) \\
                \hspace{-5pt}\tilbT_i^{k+1}
                    &\hspace{-4pt}{=}\tilbT_i^k{+}(\hatbK_i^{k+1}
                        {-}\tilbZ^{k+1})
            \end{split}
                \label{eq:admm_t Update controller synthesis}
    \end{align}
\end{subequations}
where $\barbX_i{=}(\bX_i,\hatbX_i,\hatbK_i)$, $\barbZ{=}(\bZ,\hatbZ,\tilbZ)$. 
Here, $\bZ{=}(\bZ_i)_{i=1}^N$ and $\hatbZ=(\hatbZ_i)_{i=1}^N$.
$\bT_i$, $\hatbT_i$, and $\tilbT_i$ are dual tuples corresponding to each agent.
\cref{alg:01} describes the full iterative process of using \cref{eq:admm_controller_synthesis}. 
The solution of \cref{eq:admm_x Update controller synthesis} without a penalty term can be used as the initial points $\barbX_i^0$.
$\barbZ{=}(\bX_i^0,\hatbX_i^0,\hatbK_i^0)$, $\bT_i^0{=}\hatbT_i^0{=}(\bzero,\bzero,\bzero)$, and $\tilbT_i^0{=}\bzero$ can be used as the initial consensus and dual variables, respectively. 

Since \cref{eq:Main Problem} is a nonconvex problem, \cref{eq:admm_controller_synthesis} does not guarantee convergence.
Even when it does converge, it converges to a local optimal point \cite{boyd2011distributed}.
However, as noted in \cite{lin2013design}, \gls{admm} works well when $\rho$ is sufficiently large. We have observed that a local closed-loop stability constraint can be incorporated in \cref{eq:admm_x Update controller synthesis} to improve the convergence rate.It is $(\bX_i,\hatbX_i){\in}\bbS_i$, where
\begin{align} \label{eq:closed_loop_stability_constraint}
    \bbS_i=\set*{(\bX_i,\hatbX_i)\,\bigg|\,\begin{bmatrix}
        \bQ_i+\hatbR_i & \bS_i+\hatbS_i^T \\
        \bS_i^T+\hatbS_i & \bR_i+\hatbQ_i
    \end{bmatrix}\prec0}.\;\;%\forall i\in\bbN_N.
\end{align}

% \begin{remark}
%     When calculating \cref{eq:admm_x Update controller synthesis} for different agents, $\hatbK_a$ of $a^\thh$ agents and $\hatbK_b$ of $b^\thh$ agents are different even though there are overlapped variables, $\bK_c$, between them in \cref{eq:Main Problem}.
%     Nevertheless, these two different values converge to the same value as \gls{admm} converges, since $\hatbK_a=\tilbZ=\hatbK_b$ after the convergence.
% \end{remark}

% \begin{remark}
%     To improve the convergence rate, a local closed-loop stability constraint can be incorporated in \cref{eq:admm_x Update controller synthesis}, which is $(\bX_i,\hatbX_i){\in}\bbS_i$, where
%     \begin{align} \label{eq:closed_loop_stability_constraint}
%         \bbS_i=\set*{(\bX_i,\hatbX_i)\,\bigg|\,\begin{bmatrix}
%             \bQ_i+\hatbR_i & \bS_i+\hatbS_i^T \\
%             \bS_i^T+\hatbS_i & \bR_i+\hatbQ_i
%         \end{bmatrix}\prec0},\;\;%\forall i\in\bbN_N.
%     \end{align}
% \end{remark}

\begin{algorithm} [tbp]
    \caption{Distributed Controller Synthesis}\label{alg:01}
    \begin{algorithmic}[1]
        \Require $\bX_i^0,\hatbX_i^0,\hatbK_i^0,\bZ_i^0,\hatbZ_i^0,\bT_i^0,\hatbT_i^0$ for $i\in\bbN_N$, $\epsilon_p$, and $\epsilon_d$
        \State Initialize $k=0$
        \Ensure $\hatbK_i^k$
        \While {$\bX\notin\bbQ$ or $r_p>\epsilon_p$ or $r_d>\epsilon_d$}
            \State $k\gets k+1$
            \State Find $\bX_i^k,\hatbX_i^k,\hatbK_i^k$ by \cref{eq:admm_x Update controller synthesis} in parallel
            \State Find $\bZ^k,\hatbZ^k,\tilbZ^k$ by \cref{eq:admm_z Update controller synthesis} at a centralized node
            \State Find $\bT_i^k,\hatbT_i^k,\tilbT_i^k$ by \cref{eq:admm_t Update controller synthesis} in parallel
        \EndWhile
        \State $\hatbK_i^\star=\hatbK_i^k$ for all $i\in\bbN_N$
    \end{algorithmic}
\end{algorithm}

\subsection{\texorpdfstring{$\bX_i$}{Xi} Update by \texorpdfstring{\gls{ico}}{ico}}
Directly implementing \cref{alg:01} is computationally intractable, as step 4 involves a nonconvex problem.
While various approaches could be employed, this work leverages \gls{ico} introduced in \cref{subChap:convex_overbounding} since it effectively mitigates the conservatism of solving a nonconvex problem, if initial feasible solutions are accessible. Conveniently, these are provided by the preceding \gls{admm} iteration. 
Specifically, at the $k^\thh$ \gls{admm} iteration, we set $\hatbK_i^k$, $\bX_i^k$, and $\hatbX_i^k$ as the starting points, $\hatbK_i^{k,0}$, $\bX_i^{k,0}$, and $\hatbX_i^{k,0}$, respectively, and then apply \gls{ico} following procedures in \cref{subChap:convex_overbounding}.
At each $l^\thh$ iteration of \gls{ico}, the high-order matrix inequalities corresponding to $\bbJ_i$ and $\bbK_i$ are tightened into \gls{lmi} constraints $\bbJ_i^{k,l}$ and $\bbK_i^{k,l}$, by repeatedly applying \cref{thm:overbounding} with the current iterates $\hatbK_i^{k,l}$ and $\hatbX_i^{k,l}$.
Consequently, the convex optimization problem solved at the $l^\thh$ iteration is
\begin{subequations} \label{eq:ico_x Update}
    \begin{align}
        \hspace{-5pt}\arg\!\!\!\!\min_{\substack{
                \scriptscriptstyle\!\!\!\!\!\!\!\!\!\!\!\!\dhatbK_i, \dbX_i,\dhatbX_i
            }}\;&J_i(\dhatbK_i)
            +\frac{\rho}{2}\brackets{\hspace{-5pt}\begin{array}{l}
                \|\bX_i^{k,l}{+}\dbX_i-\bZ_i^k{+}\bT_i^k\|_F^2 \\
                \;\;{+}\|\hatbX_i^{k,l}{+}\dhatbX_i{-}\hatbZ_i^k{+}\hatbT_i^k\|_F^2  \\
                \;\;{+}\|\hatbK_i^{k,l}{+}\dhatbK_i{-}\tilbZ^k{+}\tilbT_i^k\|_F^2 
            \end{array}
            \hspace{-5pt}},\hspace{-5pt} \label{eq:ico_x update_objective} \\[-2pt]
        \begin{split}
        \hspace{-5pt}\sthat\;\;\;\;&
        \dhatbK_i{\in}\bbJ_i^{k,l},\;\;
        \bX_i^{k,l}{+}\dbX_i{\in}\bbG_i,\;\;
        (\dhatbK_i,\dhatbX_i){\in}\bbK_i^{k,l}.    
        \end{split}\hspace{-5pt}
         \label{eq:ico_x Update constraint} 
    \end{align}
\end{subequations}
This \gls{ico} procedure, which serves as the internal solver for step 4 of \cref{alg:01}, is summarized in \cref{alg:02}.

\begin{algorithm} [tbp]
    \caption{Iterative Convex Overbounding}\label{alg:02}
    \begin{algorithmic}[1]
        \Require $\bX_i^k$, $\hatbX_i^k$, $\hatbK_i^k$ and $\epsilon_d$
        \Ensure $\bX_i^{k+1}$, $\hatbX_i^{k+1}$, $\hatbK_i^{k+1}$
        \State Initialize $l{=}0$, $\bX_i^{k,l}{=}\bX_i^k$, and $\hatbX_i^{k,l}{=}\hatbX_i^k$, $\hatbK_i^{k,l}{=}\hatbK_i^k$
        \While {$\frac{|J_i(\hatbK_i^{k,l})-J_i(\hatbK_i^{k,l-1})|}{|J_i(\hatbK_i^{k,l})|}\geq\epsilon$}
            \State $l\gets l{+}1$
            \State Construct constraints $\bbJ_i^{k,l}$ and $\bbK_i^{k,l}$
            \State Find $\dbX^\star$, $\dhatbX^\star$, $\dhatbK^\star$ by \cref{eq:ico_x Update}
            \State $\bX_i^{k,l}\hspace{-2pt}{=}\bX_i^{k,l{-}1}\hspace{-4pt}{+}\dbX^\star$, $\hatbX_i^{k,l}\hspace{-2pt}{=}\hatbX_i^{k,l{-}1}\hspace{-4pt}{+}\dhatbX^\star$, $\hatbK_i^{k,l}\hspace{-2pt}{=}\hatbK_i^{k,l{-}1}\hspace{-4pt}{+}\dhatbK^\star$
        \EndWhile
        \State $\bX_i^{k+1}{=}\bX_i^{k,l}$, $\hatbX_i^{k+1}{=}\hatbX_i^{k,l}$, $\hatbK_i^{k+1}{=}\hatbK_i^{k,l}$
    \end{algorithmic}
\end{algorithm}

\subsection{Initialization}
Although \gls{admm} provides warm-start points for \gls{ico} during iterative process, an initial feasible solution is still required to commence \cref{alg:01}.
The technique in \cite[Section 6]{locicero2025dissipativity} can be used if all agents are inherently open-loop stable.
Otherwise, $\scrK_i$ can be determined by the standard optimal control methods or the iterative relaxation approach in \cite{warner2017iterative}.
An additional advantage of the proposed distributed framework is the relaxation of initialization requirements.
In \cref{eq:Main Problem}, one must identify the initial feasible points satisfying all constraints including a global constraint, which is a numerically daunting task.
In contrast, our approach only requires locally feasible points.

% The proposed approach additionally gives specific advantage in finding initial feasible solutions, as it requires to find the feasible points satisfying only local constraints.
% If we solve \cref{eq:Main Problem} directly, the initial feasible points satisfying the global constraint is necessary, which is much harder to find the initial feasible points only satisfying local constraints.

\section{Application: Full-State Feedback}
As an example, the distributed synthesis method is applied to a network of agents with heterogeneous performance specifications, specifically either minimizing the $\cH_2$ or $\cH_\infty$ norm measured from disturbances to their respective linearized dynamics under full-state feedback. For this application, we assume noise-free measurements, where $\hatbn{=}\bzero$ for all $i{\in}\bbN_N$.

\subsection{Agent Level Linearized Dynamics}
In this case, each agent uses its linearized dynamics $\scrG_i^{lti}$ and full-state feedback controller $\scrK_i$, represented as
\begin{align*}
    \scrG_i^{lti}:\dbx_i{=}\bA_i\bx_i{+}\bB_i\bu_i,\;\by_i=\bx_i,\quad
    \scrK_i:\hatby_i{=}-\bK_i\hatbu_i.
\end{align*}
Given the network in \cref{eq:io_global_network}, the closed-loop dynamics of $\scrG_i^{lti}$ and $\scrK_i$ follows the following lemma.
\begin{lemma}
    Under the global network configuration in \cref{eq:io_global_network}, the transfer function from the disturbance input $\bn$ to the output $[\by_i^T\;\hatby_i^T]$ for the $i^\thh$ closed-loop subagent is given by the state-space realization $(\bA_i^{cl}{,}\bB_i^{cl}{,}\bC_i^{cl}{,}\bzero)$, where
    \begin{align} \label{eq:closed-loop parameters}
        \hspace{-5pt}\bA_i^{cl}{=}\bA_i{-}\hspace{-3pt}\sum_{j=1}^N\hspace{-2pt}\bB_i(\hspace{-.75pt}\tilbH\hspace{-.75pt})_{i\hspace{-1pt}j}\bK_j\hspace{-1pt}(\hspace{-.75pt}\hatbH\hspace{-.75pt})\hspace{-1.5pt}_{ji},
        \bB_i^{cl}{=}\bB_i(\hspace{-.75pt}\tilbH\hspace{-.75pt})_i,
        \bC_i^{cl}{=}\hspace{-2pt}\begin{bmatrix}
            \bI \\ \hspace{-1pt}{-}\bK_i(\hspace{-.75pt}\hatbH\hspace{-.75pt})_{ii}\hspace{-1pt}
        \end{bmatrix}\hspace{-2pt}{,}\hspace{-5pt}
    \end{align}
    and $(\tilbH)_i$ denotes block $i^\thh$ row of $\tilbH$.
\end{lemma}
\begin{proof}
    The global closed-loop dynamics of the interconnected system are described by
    \begin{align*}
        \dbx{=}(\bA_d{+}\bB_d\bH{-}\bB_d\tilbH\bK_d\hatbH)\bx{+}\bB_d\tilbH\bn,\;\;\;\;
        \by{=}\bx,\;\;\;\;
        \hatby{=}{-}\bK_d\hatbH\bx,
    \end{align*}
    where $\bA_d{=}\diag(\bA_i)_{i=1}^N$, and $\bB_d$, $\bK_d$, and $\bC_d$ are defined analogously.
    Since $(\bH)_{ii}{=}\bzero$, $i^\thh$ agent follows
    \begin{align*}
        \dbx_i=&\Big(\bA_i-\bB_i\hspace{-5pt}\sum_{j\in\bbN_N}\hspace{-5pt}(\tilbH)_{ij}\bK_j(\hatbH)_{ji}\Big)\bx_i+\bB_i(\tilbH)_i\bn \\[-3pt]
            &+\bB_i\hspace{-10pt}\sum_{j\in\bbN_N,j\neq i}\hspace{-10pt}(\bH)_{ij}\bx_j
            -\bB_i\hspace{-15pt}\sum_{j,k\in\bbN_N,k\neq i}\hspace{-13pt}(\tilbH)_{ij}\bK_j(\hatbH)_{jk}\bx_k \\
        \by_i{=}&\bx_i,\qquad\hatby_i{=}{-}\bK_i(\hatbH)_{ii}\bx_i{-}\bK_i\hspace{-10pt}\sum_{j\in\bbN_N,j\neq i}\hspace{-10pt}(\hatbH)_{ij}\bx_j.
    \end{align*}
    % \begin{align*}
    %     \dbx_i=&\Big(\bA_i-\bB_i\hspace{-5pt}\sum_{j\in\bbN_N}\hspace{-5pt}\tilbH_{ij}\bK_j\hatbH_{ji}\Big)\bx_i+\bB_i(\tilbH)_i\bn \\[-3pt]
    %         &+\bB_i\hspace{-10pt}\sum_{j\in\bbN_N,j\neq i}\hspace{-10pt}\bH_{ij}\bx_j
    %         -\bB_i\hspace{-15pt}\sum_{j,k\in\bbN_N,k\neq i}\hspace{-13pt}\tilbH_{ij}\bK_j\hatbH_{jk}\bx_k \\
    %     \by_i{=}&\bx_i,\qquad\hatby_i{=}{-}\bK_i\hatbH_{ii}\bx_i{-}\bK_i\hspace{-10pt}\sum_{j\in\bbN_N,j\neq i}\hspace{-10pt}\hatbH_{ij}\bx_j.
    % \end{align*}
    Therefore, the effective closed-loop parameters considering the mapping from $\bn$ to outputs follows \cref{eq:closed-loop parameters}.
    % Since the system is \gls{lti}, the transfer function from $\bn$ to the output is defined by parameters in \cref{eq:closed-loop parameters}.
\end{proof}

Since the transfer function depends on $\bK_d$, the local design parameter tuple $\hatbK_i$ is equivalent to $(\bK_i)_{i=1}^N$.
From \cref{lem:KYP Lemma},
% \begin{align} \label{eq:ctrl_kyp_constraint}
%     \bbK_i{=}\set{\hspace{-1pt}(\hspace{-1pt}\hatbK_i{,}(\hatbQ_i{,}\hatbS_i{,}\hatbR_i\hspace{-1pt})\hspace{-1pt})|{-}\hatbR_i{+}\hatbS_i^T\bK_i{+}\bK_i^T\hatbS_i{-}\bK_i^T\hatbQ_i\bK_i{\preceq}0},\hspace{-5pt}
% \end{align}
\begin{align} \label{eq:ctrl_kyp_constraint}
    \bbK_i{=}\set{(\hatbK_i{,}\hatbX_i)|{-}\hatbR_i{+}\hatbS_i^T\bK_i{+}\bK_i^T\hatbS_i{-}\bK_i^T\hatbQ_i\bK_i{\preceq}0},
\end{align}
For an agent $i$ targeting $\cH_2$-norm minimization, $J_i(\hatbK_i)=\tr(\bW_i)$, where $\bW_i{\succ}0$ is a slack variable satisfying
\begin{align} \label{eq:H2_constraint}
    \hspace{-5pt}\bbJ_i{=}\set*{\hspace{-2pt}\hatbK_i\Big|{\exists}\bP_i{\succ}0\,\sthat
    {\setlength{\arraycolsep}{0pt}\begin{bmatrix}
        \he(\bP_i\bA_i^{cl}) & * \\
        \bC_i^{cl} & {-}\bI
    \end{bmatrix}}
    {\preceq}0{,}(\hspace{-.5pt}\bB_i^{cl}\hspace{-1pt})^T\hspace{-1pt}\bP_i\bB_i^{cl}{\prec}\bW_i
    \hspace{-1pt}}.\hspace{-5pt}
\end{align}
Conversely, for an agent $i$ targeting $\cH_\infty$-norm minimization, $J_i(\hatbK_i){=}\gamma_i$ from some $\gamma_i{\geq}0$ such that
\begin{align}\label{eq:Hinf_constraint}
    \bbJ_i{=}\set*{\hspace{-1pt}\hatbK_i\,\Big|\,{\exists}\bP_i{\succ}0\,\sthat
    {\setlength{\arraycolsep}{3pt}\begin{bmatrix}
        \he(\bP_i\bA_i^{cl}) & * & * \\
        (\bP_i\bB_i^{cl})^T & {-}\gamma_i\bI & * \\
        \bC_i^{cl} & \bzero & -\gamma_i\bI
    \end{bmatrix}}
    {\prec}0
    \hspace{-1pt}}.
\end{align}

\subsection{Tightening Constraints}
\cref{cor:lmi_kyp,cor:lmi_H2,cor:lmi_Hinf} provide the tight constraints of $\bbK_i$ and $\bbJ_i$. Next, we use $\bA_i^{cl,0}{=}\bA_i{-}\sum_{j=1}^N\bB_i(\tilbH)_{ij}\bK_j^0(\hatbH)_{ji}$, $\dbA_i^{cl}{=}{-}\sum_{j=1}^N\bB_i(\tilbH)_{ij}\dbK_j(\hatbH)_{ji}$, $\bC_i^{cl,0}{=}[\bI\;{-}(\bK_i^0(\bH)_{ii})^T]^T$ and $\dbC_i^{cl}=[\bzero\;{-}(\dbK_i(\bH)_{ii})^T]^T$.
\begin{corollary} \label{cor:lmi_kyp}
    Given $\hatbK_i^0,\hatbQ_i^0,\hatbS_i^0$, and $\hatbR_i^0$, if there exist $\dhatbK_i$, $\dhatbQ_i,\dhatbS_i$, and $\dhatbR_i$ such that $(\dhatbK_i{,}\hatbX_i){\in}\bbK_i^c$, where 
    \begin{align} \label{eq:lmi_kyp}
    \hspace{-5pt}\bbK_i^c{=}\set*{\hspace{-7pt}\begin{array}{l}
            (\dhatbK_i,\dhatbX_i)\big| \\
            \;{\setlength{\arraycolsep}{1pt}\begin{bmatrix}
                {-}\cT_0^1(\hatbR_i{-}\he(\hatbS_i^T\hspace{-1pt}\bK_i){+}\bK^T_i\hspace{-1pt}\hatbQ_i\bK_i) & * & * \\
                \dhatbS_i{-}\frac{1}{2}\hatbQ_i^0\dbK_i{+}\dbK_i & \hspace{-3pt}-2\bI & * \\
                {-}\dhatbQ_i\bK_i^0+\dbK & \hspace{-10pt}{-}\frac{1}{2}\dhatbQ_i & -2\bI
            \end{bmatrix}
            }{\preceq}0
        \end{array}\hspace{-7pt}
        },\hspace{-5pt}
    \end{align}
    then $(\hatbK_i^0{+}\dhatbK_i,\hatbX_i^0{+}\dhatbX_i){\in}\bbK_i$. Moreover, \cref{eq:lmi_kyp} is feasible if $\hatbK_i^0$, $\hatbQ_i^0$, $\hatbS_i^0$, and $\hatbR_i^0$ are feasible for \cref{eq:ctrl_kyp_constraint}.
\end{corollary}
\begin{proof}
    The proof follows by applying \cref{eq:overbounding_sebe} of \cref{thm:overbounding} sequentially with $\bG{=}\bI$. First, use $\bK_i{=}\bK_i^0{+}\dbK_i$ and $\hatbS_i{=}\hatbS_i^0{+}\dhatbS_i$. Next, use $\hatbQ_i{=}\hatbQ_i^0{+}\dhatbQ_i$.
\end{proof}

\begin{corollary} \label{cor:lmi_H2}
    Given $\hatbK_i^0$ and $\bP_i^0$, if there exist $\dhatbK_i$ and $\dbP_i$ such that $\dhatbK_i{\in}\bbJ_i^c$, where 
    \begin{align} \label{eq:lmi_H2}
    \hspace{-5pt}\bbJ_i^c{=}\set*{\hspace{-7pt}\begin{array}{l}
            \dhatbK_i\,\big|\,\exists\dbP_i\;\sthat\;\bP_i^0{+}\dbP_i{\succ}0, \\
            \hspace{-2pt}{\setlength{\arraycolsep}{1pt}\begin{bmatrix}
            \hspace{-1.5pt}\he(\cT_0^1(\bP_i\bA_i^{cl})\hspace{-1.5pt}) & * & * \\
            \bC_i^{cl,0}{+}\dbC_i^{cl} & \hspace{-5pt}-\bI & \bzero \\
            \dbP_i{+}\dbA_i^{cl} & \hspace{-1pt}\bzero & \hspace{-3pt}-2\bI\hspace{-1.5pt}
        \end{bmatrix}}\hspace{-2pt}{\preceq}0,\hspace{-2pt}
        {\setlength{\arraycolsep}{1pt}\begin{bmatrix}
            {-}\bW_i & * \\
            \hspace{-1.5pt}\bP_i^0\bB_i^{cl} & {-}\bP_i^0{+}\dbP_i\hspace{-1.5pt}
        \end{bmatrix}}\hspace{-2pt}{\preceq}0
        \end{array}\hspace{-7pt}
        },\hspace{-5pt}
    \end{align}
    then $\hatbK_i^0{+}\dhatbK_i$ satisfies \cref{eq:H2_constraint}. Moreover, \cref{eq:lmi_H2} is feasible if $\hatbK_i^0$ is feasible for \cref{eq:H2_constraint}.
\end{corollary}
\begin{proof}
    The first \gls{lmi} follows by applying \cref{eq:overbounding_sebe} of \cref{thm:overbounding} with $\bG{=}\bI$, $\bP_i{=}\bP_i^0{+}\dbP$ and $\bA_i^{cl}{=}\bA_i^{cl,0}{+}\dbA_i^{cl}$.
    The second \gls{lmi} follows applying Schur complement, using that ${-}\bP_i^{{-}1}\hspace{-2pt}{\prec}{-}2(\bP_i^0)^{-1}{+}(\bP_i^0)^{-1}\bP_i(\bP_i^0)^{-1}$, and then pre and post multiplication of $\diag(\bI,\bP_i^0)$.
\end{proof}

\begin{corollary} \label{cor:lmi_Hinf}
    Given $\hatbK_i^0$, and $\bP_i^0$, if there exist $\dhatbK_i$, $\dbP_i$, and $\gamma_i{\geq}0$ such that $\dhatbK_i{\in}\bbJ_i^c$, where 
    \begin{align} \label{eq:lmi_Hinf}
    \bbJ_i^c{=}\set*{\hspace{-5pt}
            \begin{array}{l}
                \dhatbK_i|\exists\dbP_i\;\sthat\;\bP_i^0{+}\dbP_i{\succ}0, \\
                \begin{bmatrix}
                \he(\cT_0^1(\bP_i\bA_i^{cl}) & * & * & * \\
                ((\bP_i^0{+}\dbP_i)\bB_i^{cl})^T & {-}\gamma_i\bI & \bzero & \bzero \\
                \bC_i^{cl,0}{+}\dbC_i^{cl} & \bzero & {-}\gamma_i\bI & \bzero \\
                \dbP_i{+}\dbA_i^{cl} & \bzero & \bzero & -2\bI
            \end{bmatrix}{\prec}0
            \end{array}\hspace{-5pt}
            % \dhatbK_i\Big|\exists\bP_i^0{+}\dbP_i\,\sthat
            % \hspace{-2pt}\setlength{\arraycolsep}{1pt}{\begin{bmatrix}
            %     \he(\cT_0^1(\bP_i\bA_i^{cl}) & * & * & * \\
            %     \hspace{-2pt}(\hspace{-1.5pt}(\bP_i^0{+}\dbP_i\hspace{-1pt})\bB_i^{cl})^T & \hspace{-2pt}{-}\gamma_i\bI\hspace{-1pt} & \bzero & \bzero \\
            %     \bC_i^{cl,0}\hspace{-2pt}{+}\dbC_i^{cl} & \bzero & \hspace{-5pt}{-}\gamma_i\bI\hspace{-1pt} & \bzero \\
            %     \dbP_i{+}\dbA_i^{cl} & \bzero & \bzero & \hspace{-3pt}-2\bI\hspace{-2pt}
            % \end{bmatrix}
            % }\hspace{-2pt}{\prec}0\hspace{-3pt}
        },%\hspace{-5pt}
    \end{align}
    then $\hatbK_i^0\hspace{-2pt}{+}\dhatbK_i$ is in \cref{eq:Hinf_constraint}. Moreover, \cref{eq:lmi_Hinf} is feasible if $\hatbK_i^0$ is feasible for \cref{eq:Hinf_constraint}.
\end{corollary}
\begin{proof}
    The proof follows by applying \cref{eq:overbounding_sebe} of \cref{thm:overbounding} with $\bG{=}\bI$, $\bP_i{=}\bP_i^0{+}\dbP$ and $\bA_i^{cl}{=}\bA_i^{cl,0}{+}\dbA_i^{cl}$.
\end{proof}

\subsection{Numerical Example}
\begin{figure}
    \centering
    \includegraphics[height = 0.065\textheight]{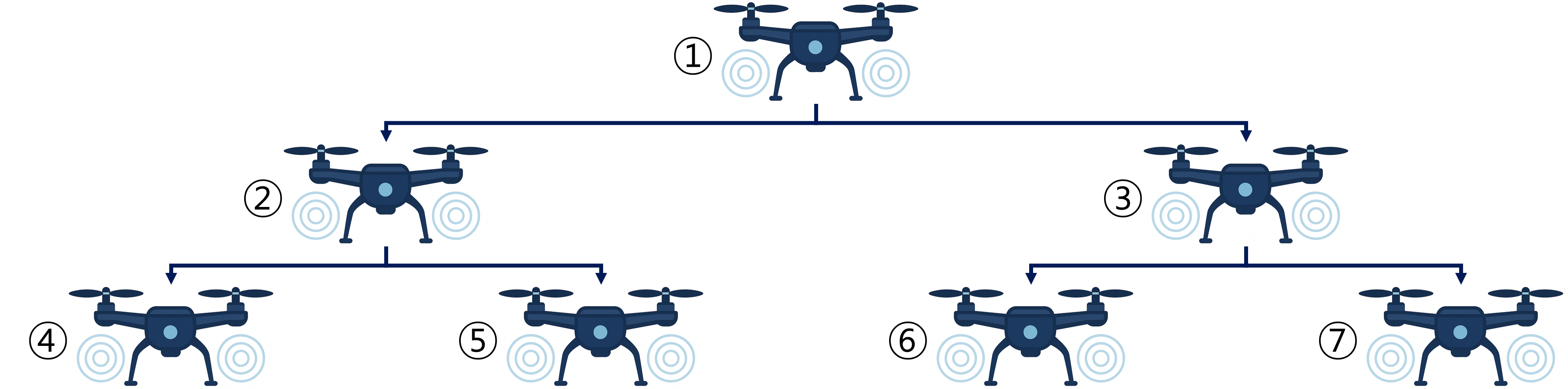}
    \caption{Network of \glspl{uav}.}
    \label{fig:Graph}
    \vspace*{-1.75\baselineskip} 
\end{figure}

We consider a network of 7 \glspl{uav} operating in a 2D plane as shown in \cref{fig:Graph}.
The parameters for each agent are assumed to be uncertain, varying uniformly within $\pm$10\% of  their nominal values.
The dynamics of the $i^\thh$ agent and the network interconnection are defined as
\begin{align*}
    \bA_i{=}{\setlength{\arraycolsep}{2.5pt}\begin{bmatrix}  
        \bzero_{3{\times} 2} & \bzero_{3{\times} 1} & \bI_3 \\ 
        \bzero_{1{\times} 2} & -g & \bzero_{1{\times} 3} \\ 
        \bzero_{2{\times}2} & \bzero_{2{\times} 1} & \bzero_{2{\times} 3}\end{bmatrix}}\hspace{-2pt},\;
    \bB_i{=}{\setlength{\arraycolsep}{2.5pt}\begin{bmatrix}
        \bzero_{4{\times} 1} & \bzero_{4{\times} 1} \\
        \frac{1}{m_i} & \frac{1}{m_i} \\
        -\frac{l_i}{I_{i}} & \frac{l_i}{I_{i}}
    \end{bmatrix}}\hspace{-2pt},\;
    \hatbH{=}{\setlength{\arraycolsep}{2.5pt}\begin{bmatrix}
        \bL_{1} & \bzero_{18\times24} \\ 
        \bL_{2} & \bI_{24}
    \end{bmatrix}}\hspace{-1pt},
\end{align*}
$\bH{=}\bzero_{21{\times}42}$, and $\tilbH{=}\bI_{21}$, where 
\begin{align*}
    \bL_1{=}{\setlength{\arraycolsep}{2.5pt}\begin{bmatrix}
        \bI_6 & \bzero_{6{\times}6} & \bzero_{6{\times}6}\\
        -\bI_6 & \bI_6 &\bzero_{6{\times}6} \\
        -\bI_6 & \bzero_{6{\times}6} & \bI_6
    \end{bmatrix}},\quad
    \bL_2{=}{\setlength{\arraycolsep}{2.5pt}\begin{bmatrix}
        \bzero_{6{\times}6} & -\bI_6 & \bzero_{6{\times}6} \\
        \bzero_{6{\times}6} & -\bI_6 & \bzero_{6{\times}6} \\
        \bzero_{6{\times}6} & \bzero_{6{\times}6} & -\bI_6 \\
        \bzero_{6{\times}6} & \bzero_{6{\times}6} & -\bI_6
    \end{bmatrix}},
    % \;
    % \bL_3{=}{\setlength{\arraycolsep}{2pt}\begin{bmatrix}
    %     \bzero_6 & \bzero_6 & -\bI_6 \\
    %     \bzero_6 & \bzero_6 & -\bI_6
    % \end{bmatrix}}.
\end{align*}
where $g$ is gravitational acceleration and $m_n$=3 kg, $l_n$=0.2m, and $I_n$=1 kg-m$^2$ are the nominal mass, moment of inertia, and wing length of all \glspl{uav}.
%The nominal values are $m_n$=3 kg, $l_n$=0.2m, and $I_n$=1 kg-m$^2$.

Each agent is a linear system with polytopic uncertainty in $\bB_i$ for all $i\in\bbN_7$, so the \gls{lmi} in \cite[Lemma 2]{jang2025communication} is used as a dissipativity constraint in \cref{eq:Main:plant}.
The $i^\thh$ closed-loop dynamics are parameterized as
\begin{align*}
    \bA_i^{cl}=\bA_i{-}\bB_i\bK_i,\;\;
    \bB_i^{cl}=\bB_i,\;\;
    (\bC_i^{cl})^T=\begin{bmatrix}
        \bI^T & -\bK_i^T
    \end{bmatrix}^T.
\end{align*}
Note that the $i^\thh$ local transfer function is independent of $\bK_j$ for $j{\neq} i$, so the consensus variable $\tilbZ$ is omitted and $\hatbK_i{=}\bK_i$.

Agents 1-3 minimize their $\cH_2$ norms, while agents 4-7 minimize their $\cH_\infty$ norms. Both objectives are calculated relative to the nominal dynamics.
The \gls{lqr} controllers with $\bQ_{lqr}{=}\bI_6$ and $\bR_{lqr}{=}\bI_3$ are used as initial feasible controllers of each \gls{uav}.
Using these initial controllers, the nominal $\cH_2$ and $\cH_\infty$ norms are $4.702$ and $1.829$, respectively.
\cref{eq:closed_loop_stability_constraint} is added to accelerate convergence.
To run \cref{alg:01,alg:02}, $\rho{=}100$ and $\epsilon{=}\epsilon_p{=}\epsilon_d=10^{-3}$ are used.
\cref{fig:performance} shows that \cref{alg:01} converges after 7468 iterations and significantly reduces $\cH_\infty$ norm of agents. Interestingly, $\cH_2$ norm remain unchanged, suggesting the initial \gls{lqr} design is already near-optimal for the $\cH_2$ objective.

We compare the performance of the proposed method against the centralized results obtained via \gls{ico} with $\epsilon{=}10^{-3}$, as summarized in \cref{tb:performance}. 
As previously noted, the $\cH_2$ norm remains unchanged.
While the centralized controller synthesized without \gls{ndt} achieves the highest performance, it fails to guarantee stability for the real system under uncertainty.
The performance of the proposed framework is only slightly lower than that of the centralized approach, and it significantly improves upon the initial guess.
Moreover, our method achieves acceptable performance without sharing agent's local dynamics across the entire network. Hence, this method improves privacy with little sacrifice in performance.

% \cref{fig:performance,tb:performance} show that \cref{alg:01} successfully converges after 7468 iterations and significantly reduces $\cH_\infty$ norm of agents. Interestingly, $\cH_2$ norm remain unchanged, suggesting the initial \gls{lqr} design is already near-optimal for the $\cH_2$ objective.

% \begin{figure}
% \centering
%     \begin{subfigure}[t]{0.15\textwidth}
%         \centering
%         \resizebox{\textwidth}{!}{\input{Figures_tex/H2_norm}}
%     \end{subfigure}
%     \begin{subfigure}[t]{0.15\textwidth}
%         \centering
%         \resizebox{\textwidth}{!}{\input{Figures_tex/Hinf_norm}}
%     \end{subfigure}
%     \begin{subfigure}[t]{0.15\textwidth}
%         \centering
%         \resizebox{\textwidth}{!}{\input{Figures_tex/total_norm}}
%     \end{subfigure}
%     \caption{System responses to \texorpdfstring{$\cL_2$}{L2} disturbances: The dotted lines indicate equilibrium points of each \gls{uav}.} \label{fig:stability_analysis_result}
%     % \vspace*{-1.25\baselineskip} 
% \end{figure}

\begin{figure}[t]
    \captionsetup[sub]{aboveskip=0pt, belowskip=0pt}
    \centering
    \subfloat[$\cH_2$ norms for first 3 \glspl{uav}]{\includegraphics[width=0.23\textwidth]{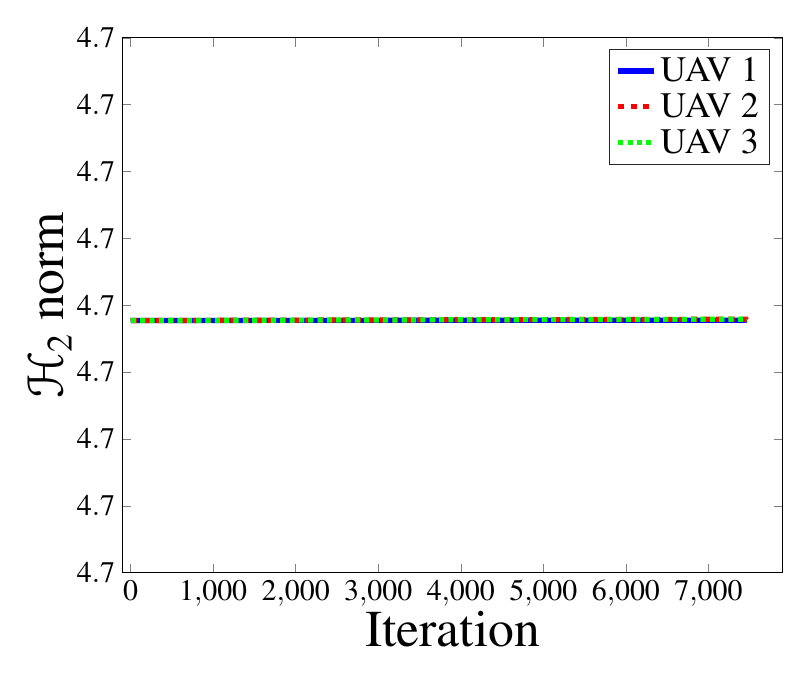}} \label{fig:H2}
    \subfloat[$\cH_\infty$ norm for last 4 \glspl{uav}]{\includegraphics[width=0.23\textwidth]{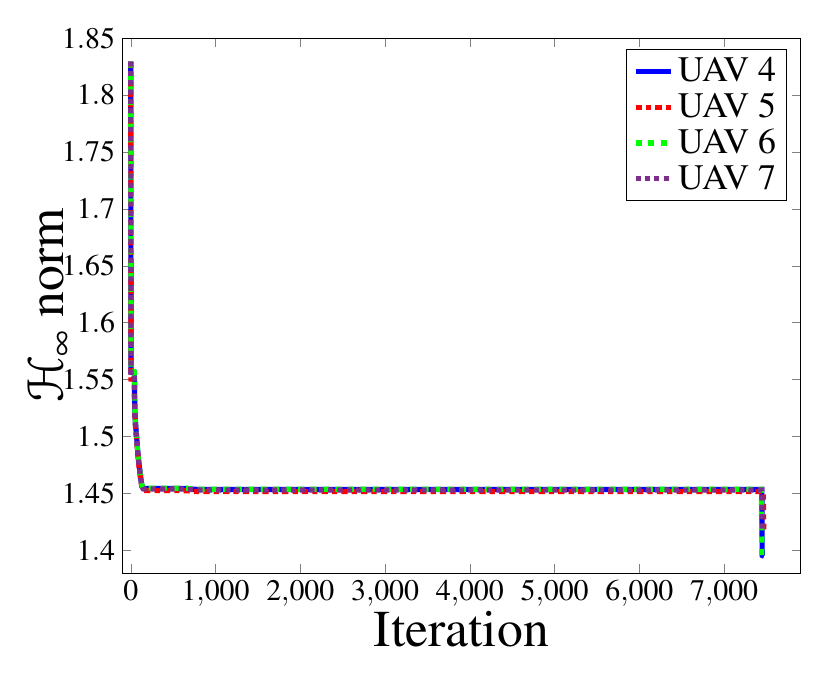}} \label{fig:Hinf}
    % \subfloat[]{\includegraphics[width=0.155\textwidth]{Figures_tex/Hnormsum.pdf}} \label{fig:sum}
    % \vspace{-8pt}
    \caption{Evolution of performance metrics over \cref{alg:01}.} \label{fig:performance}
    \vspace*{-.5\baselineskip} 
\end{figure}

\setlength{\tabcolsep}{1pt}
\begin{table} [t]
\begin{center}
\caption{Agent performance comparison. 1$^\text{st}$ row: distributed synthesis via \cref{alg:01}; 2$^\text{nd}$ row: centralized synthesis; 3$^\text{rd}$ row: centralized synthesis without \cref{eq:Main:stability}.}\label{tb:performance}
\begin{tabular}{ccccccccc}
\hline\hline
& \gls{uav}$\,$1 & \gls{uav}$\,$2 & \gls{uav}$\,$3 & \gls{uav}$\,$4 & \gls{uav}$\,$5 & \gls{uav}$\,$6 & \gls{uav}$\,$7 \\ \hline
Distributed & 4.702 & 4.702 & 4.702 & 1.395 & 1.418 & 1.395 & 1.418 \\
Centralized & 4.702 & 4.702 & 4.702 & 1.325 & 1.325 & 1.325 & 1.325 \\
Centralized (no \gls{ndt})\hspace{-1pt} & 4.702 & 4.702 & 4.702 & 1.322 & 1.322 & 1.322 & 1.322
\\ \hline\hline
\end{tabular}
\end{center}
\vspace*{-2.5\baselineskip} 
\end{table}

\begin{figure}[t]
\centering
    \begin{subfigure}[t]{0.23\textwidth}
        \centering
        \resizebox{\textwidth}{!}{% This file was created by matlab2tikz.
\definecolor{mycolor1}{rgb}{0.00000,0.44700,0.74100}%
\definecolor{mycolor2}{rgb}{0.85000,0.32500,0.09800}%
\definecolor{mycolor3}{rgb}{0.92900,0.69400,0.12500}%
\definecolor{mycolor4}{rgb}{0.49400,0.18400,0.55600}%
\definecolor{mycolor5}{rgb}{0.46600,0.67400,0.18800}%
\definecolor{mycolor6}{rgb}{0.30100,0.74500,0.93300}%
\definecolor{mycolor7}{rgb}{0.63500,0.07800,0.18400}%
\begin{tikzpicture}

\begin{axis}[%
width=4.521in,
height=3.566in,
at={(0.758in,0.481in)},
scale only axis,
xmin=0,
xmax=30,
xlabel style={font=\color{white!15!black}},
xlabel={time [s]},
ymin=-0.8,
ymax=2.8,
ylabel style={font=\color{white!15!black}},
ylabel={x position [m]},
axis background/.style={fill=white},
xmajorgrids,
ymajorgrids,
legend style={at={(0.97,0.03)}, anchor=south east, legend cell align=left, align=left, draw=white!15!black},
xlabel style={font={\Huge}},ylabel style={font=\Huge},legend style={font=\Large},
]
\addplot [color=mycolor1, line width=2.5pt]
  table[x=t, y=e, col sep=comma]
  {Figures_tex/Data/x1.csv};
\addlegendentry{UAV 1}

\addplot [color=mycolor2, line width=2.5pt]
  table[x=t, y=e, col sep=comma]
  {Figures_tex/Data/x2.csv};
\addlegendentry{UAV 2}

\addplot [color=mycolor3, line width=2.5pt]
  table[x=t, y=e, col sep=comma]
  {Figures_tex/Data/x3.csv};
\addlegendentry{UAV 3}

\addplot [color=mycolor4, line width=2.5pt]
  table[x=t, y=e, col sep=comma]
  {Figures_tex/Data/x4.csv};
\addlegendentry{UAV 4}

\addplot [color=mycolor5, line width=2.5pt]
  table[x=t, y=e, col sep=comma]
  {Figures_tex/Data/x5.csv};
\addlegendentry{UAV 5}

\addplot [color=mycolor6, line width=2.5pt]
  table[x=t, y=e, col sep=comma]
  {Figures_tex/Data/x6.csv};
\addlegendentry{UAV 6}

\addplot [color=mycolor7, line width=2.5pt]
  table[x=t, y=e, col sep=comma]
  {Figures_tex/Data/x7.csv};
\addlegendentry{UAV 7}

\addplot [color=black, dotted, line width=2pt]
  table[row sep=crcr]{%
0	0\\
30	0\\
};
% \addlegendentry{Equilibrium}

\addplot [color=black, dotted, line width=2pt]
  table[row sep=crcr]{%
0	1\\
30	1\\
};

\addplot [color=black, dotted, line width=2pt]
  table[row sep=crcr]{%
0	2\\
30	2\\
};

\end{axis}
\end{tikzpicture}%
}
    \end{subfigure}
    \begin{subfigure}[t]{0.23\textwidth}
        \centering
        \resizebox{\textwidth}{!}{% This file was created by matlab2tikz.
\definecolor{mycolor1}{rgb}{0.00000,0.44700,0.74100}%
\definecolor{mycolor2}{rgb}{0.85000,0.32500,0.09800}%
\definecolor{mycolor3}{rgb}{0.92900,0.69400,0.12500}%
\definecolor{mycolor4}{rgb}{0.49400,0.18400,0.55600}%
\definecolor{mycolor5}{rgb}{0.46600,0.67400,0.18800}%
\definecolor{mycolor6}{rgb}{0.30100,0.74500,0.93300}%
\definecolor{mycolor7}{rgb}{0.63500,0.07800,0.18400}%
\begin{tikzpicture}

\begin{axis}[%
width=4.521in,
height=3.566in,
at={(0.758in,0.481in)},
scale only axis,
xmin=0,
xmax=30,
xlabel style={font=\color{white!15!black}},
xlabel={time [s]},
ymin=-4,
ymax=4,
ylabel style={font=\color{white!15!black}},
ylabel={y position [m]},
axis background/.style={fill=white},
xmajorgrids,
ymajorgrids,
legend style={at={(0.97,0.03)}, anchor=south east, legend cell align=left, align=left, draw=white!15!black},
xlabel style={font={\Huge}},ylabel style={font=\Huge},legend style={font=\Large},
]
\addplot [color=mycolor1, line width=2.5pt]
  table[x=t, y=e, col sep=comma]
  {Figures_tex/Data/y1.csv};
\addlegendentry{UAV 1}

\addplot [color=mycolor2, line width=2.5pt]
  table[x=t, y=e, col sep=comma]
  {Figures_tex/Data/y2.csv};
\addlegendentry{UAV 2}

\addplot [color=mycolor3, line width=2.5pt]
  table[x=t, y=e, col sep=comma]
  {Figures_tex/Data/y3.csv};
\addlegendentry{UAV 3}

\addplot [color=mycolor4, line width=2.5pt]
  table[x=t, y=e, col sep=comma]
  {Figures_tex/Data/y4.csv};
\addlegendentry{UAV 4}

\addplot [color=mycolor5, line width=2.5pt]
  table[x=t, y=e, col sep=comma]
  {Figures_tex/Data/y5.csv};
\addlegendentry{UAV 5}

\addplot [color=mycolor6, line width=2.5pt]
  table[x=t, y=e, col sep=comma]
  {Figures_tex/Data/y6.csv};
\addlegendentry{UAV 6}

\addplot [color=mycolor7, line width=2.5pt]
  table[x=t, y=e, col sep=comma]
  {Figures_tex/Data/y7.csv};
\addlegendentry{UAV 7}

\addplot [color=black, dotted, line width=2.0pt]
  table[row sep=crcr]{%
0	0\\
30	0\\
};
% \addlegendentry{Equilibrium}

\addplot [color=black, dotted, line width=2.0pt]
  table[row sep=crcr]{%
0	1\\
30	1\\
};

\addplot [color=black, dotted, line width=2.0pt]
  table[row sep=crcr]{%
0	2\\
30	2\\
};

\addplot [color=black, dotted, line width=2.0pt]
  table[row sep=crcr]{%
0	3\\
30	3\\
};

\addplot [color=black, dotted, line width=2.0pt]
  table[row sep=crcr]{%
0	-1\\
30	-1\\
};

\addplot [color=black, dotted, line width=2.0pt]
  table[row sep=crcr]{%
0	-2\\
30	-2\\
};

\addplot [color=black, dotted, line width=2.0pt]
  table[row sep=crcr]{%
0	-3\\
30	-3\\
};

\end{axis}
\end{tikzpicture}%
}
    \end{subfigure}
    \caption{System responses to \texorpdfstring{$\cL_2$}{L2} disturbances. The dotted lines indicate equilibrium points of each \gls{uav}.} \label{fig:response}
    \vspace*{-1.25\baselineskip} 
\end{figure}
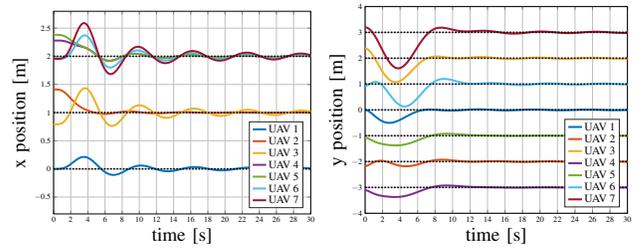

\begin{figure}[t]
\centering
    \includegraphics[width=0.30\textwidth]{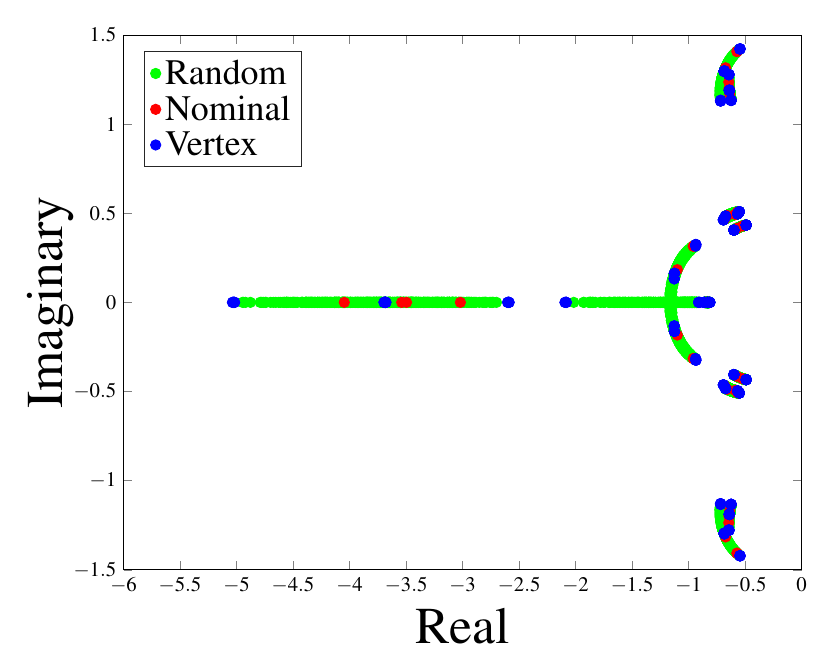}
    \caption{Pole locations of $\diag(\bA_i^{cl})_{i\in\bbN_7}$. ``Random" denotes poles from sampled realizations, ``Vertex" denotes poles corresponding to the 8 vertices of the polytopic uncertainty.} \label{fig:pole}
    \vspace*{-1.75\baselineskip} 
\end{figure}

Each agent input incurs an $\cL_2$ disturbance, $\bn_i{=}[e^{-t}\; \frac{\sin t}{t}]^T$ for all $i{\in}\bbN_7$ and its initial states are randomly sampled from a centered interval of $\pm0.5$ m relative to its equilibrium.
\cref{fig:response} shows that all agents asymptotically return to their equilibrium states, which validates that the entire network is $\cL_2$ stable.
In addition, we calculate the eigenvalues of global $\bA$ matrix for 100 sampled physical parameters and all 8 vertices of polytopic uncertainty.
As shown in \cref{fig:pole}, all poles remain the left-half plane.
This is ensured by the dissipativity-based design.

\section{CONCLUSIONS} \label{chap:Conclusion}
This paper presents a distributed controller synthesis method that preserves the privacy of each agent's dynamics.
The proposed algorithm provides a framework in which each agent independently synthesize a local controller while satisfying a global consensus condition associated with the network-wide stability.
As one possible application, the framework is applied to the system with full-state measurement.
The 2D swarm \glspl{uav} example shows that the proposed framework successfully synthesizes local controllers without sharing agent's dynamical information.

\addtolength{\textheight}{-12cm}   % This command serves to balance the column lengths

%%%%%%%%%%%%%%%%%%%%%%%%%%%%%%%%%%%%%%%%%%%%%%%%%%%%%%%%%%%%%%%%%%%%%%%%%%%%%%%%

%%%%%%%%%%%%%%%%%%%%%%%%%%%%%%%%%%%%%%%%%%%%%%%%%%%%%%%%%%%%%%%%%%%%%%%%%%%%%%%%

%%%%%%%%%%%%%%%%%%%%%%%%%%%%%%%%%%%%%%%%%%%%%%%%%%%%%%%%%%%%%%%%%%%%%%%%%%%%%%%%

% \section*{ACKNOWLEDGMENT}

% TBD

%%%%%%%%%%%%%%%%%%%%%%%%%%%%%%%%%%%%%%%%%%%%%%%%%%%%%%%%%%%%%%%%%%%%%%%%%%%%%%%%

\bibliographystyle{IEEEtran}
\bibliography{MyBib}{}

\end{document}